\documentclass{article}
\usepackage{ijcai25}

\usepackage{times}
\usepackage{soul}
\usepackage{url}
\usepackage[hidelinks]{hyperref}
\usepackage[utf8]{inputenc}
\usepackage[small]{caption}
\usepackage{graphicx}
\usepackage{amsmath}
\usepackage{amsthm}
\usepackage{booktabs}
\usepackage{algorithm}
\usepackage[switch]{lineno}
\usepackage{amsmath}
\usepackage{bm}
\usepackage{framed}
\usepackage{mdframed}
\usepackage{lipsum}
\usepackage{algpseudocode}
\usepackage{mathtools}
\usepackage{amsfonts}
\usepackage{mathrsfs}
\usepackage{graphicx} 
\usepackage{comment}
\usepackage{balance}
\usepackage{appendix}
\usepackage{subfigure}
\usepackage{stfloats}
\usepackage{comment} 
\usepackage{tcolorbox}
\long\def\comment#1{}

\newtheorem{example}{Example}
\newtheorem{theorem}{Theorem}
\newtheorem{definition}{Definition}
\newtheorem{proposition}{Proposition}
\newtheorem{assumption}{Assumption}
\newtheorem{corollary}{Corollary}

\title{Task Allocation in Customer-led Two-sided Markets\\ with Satellite Constellation Services}

\author{Jianglin Qiao$^1$\and Zehong Cao$^1$\and Dave de Jonge$^{2}$\and Ryszard Kowalczyk$^{1,3}$\\
\affiliations
$^1$University of South Australia\\
$^2$IIIA-CSIC\\
$^3$Systems Research Institute, Polish Academy of Sciences\\
\emails
\{Jianglin Qiao, Jimmy.Cao, Ryszard.Kowalczyk\}@unisa.edu.au\and
davedejonge@iiia.csic.es
}
\begin{document}
\maketitle

\begin{abstract}
Multi-agent systems (MAS) are increasingly applied to complex task allocation in two-sided markets, where agents such as companies and customers interact dynamically. Traditional company-led Stackelberg game models, where companies set service prices, and customers respond, struggle to accommodate diverse and personalised customer demands in emerging markets like crowdsourcing. This paper proposes a customer-led Stackelberg game model for cost-efficient task allocation, where customers initiate tasks as leaders, and companies create their strategies as followers to meet these demands. We prove the existence of Nash Equilibrium for the follower game and Stackelberg Equilibrium for the leader game while discussing their uniqueness under specific conditions, ensuring cost-efficient task allocation and improved market performance. Using the satellite constellation services market as a real-world case, experimental results show a 23\% reduction in customer payments and a 6.7-fold increase in company revenues, demonstrating the model's effectiveness in emerging markets.
\end{abstract}

\section{Introduction}

Two-sided markets are economic models where companies and customers depend on each other \cite{banerjee2017segmenting}. Companies provide services, while customers generate demand and revenue \cite{patro2020incremental}. In these markets, companies’ pricing, resource allocation, and service quality influence customer decisions, while customer preferences and willingness to pay shape company strategies. Effective task allocation is key to balancing supply and demand, ensuring resources meet customer needs while maximising company profits, all within the constraints of resources, pricing, and customer diversity.

Two-sided market model design typically draws on game-based and non-game-based approaches. Non-game-based approaches aim to optimise market mechanisms by focusing on pricing strategies and overall market efficiency. These methods enhance performance by refining market operations without explicitly modelling the interactions between companies and customers \cite{lee2023distributed}. In contrast, game theory-based approaches provide a rigorous framework for analysing decision-making in multi-agent environments, capturing the competitive and cooperative dynamics between companies and customers \cite{bataineh2021cloud}. Among the game-based approaches, the Stackelberg Game model is a prominent tool for two-sided markets \cite{xu2016signaling}. 

\begin{figure}[ht]
    \centering
    \includegraphics[scale=0.18]{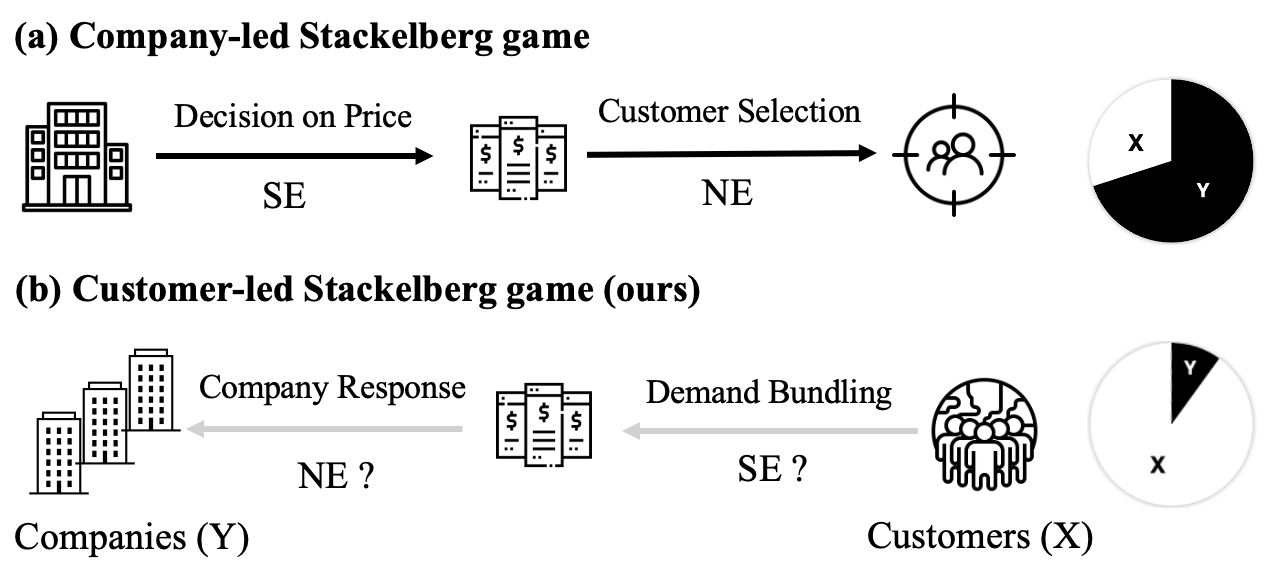}
    \caption{(a) Company-led vs (b) Customer-led Stackelberg game approach for Two-sided Market}\label{different}
\end{figure}

In the company-led Stackelberg model, platforms or companies assume the role of leaders, setting prices or service strategies, while customers act as followers, responding to these decisions. This structure facilitates the identification of optimal strategy, where both companies and customers benefit by reaching a Stackelberg equilibrium (SE) in the leader game and a Nash Equilibrium in the follower game, improving market stability and efficiency. \textbf{Fig.~\ref{different}(a)} illustrates this classical structure, where companies determine prices and customers select providers, resulting in SE and NE \cite{li2018integrating,wang2022coordinating}. However, with the increasing complexity of market demands and the growing number of customers, the limitations of company-led models have become apparent. This shift has given rise to customer-led models in two-sided markets, as shown in \textbf{Fig.~\ref{different}(b)}. In a customer-led model, customers initiate tasks, and companies, now acting as followers, adapt their strategies in response to customers' diverse and personalised demands. This inversion of roles introduces a critical research gap in achieving SE and NE, as existing models do not fully account for the dynamic adjustments companies require to meet evolving customer preferences. The complexity arising from this role reversal makes the existence and stability of such equilibrium less certain, highlighting the need for new theoretical frameworks and adaptive solution approaches to address this gap.

In many customer-led, two-sided markets, such as cloud computing services \cite{du2022sdn}, ride-sharing platforms \cite{mo2023stochastic}, and crowdsourcing platforms \cite{kang2023incentive}, current multi-agent systems (MAS) face challenges in providing effective task allocation solutions. The complexity of these markets lies in high operational costs, resource constraints, and the growing demand for personalised services, emphasising the need for adaptive task allocation strategies. The satellite constellation services market serves as a representative example of a customer-led, two-sided market. As of the end of 2024, there are 1,800 Earth Observation satellites operated by over 250 satellite constellation service providers, 70\% of which are small-scale providers owning only 1-2 satellites and serving tens of thousands of users \cite{NovaSpace2025}. Insights from our recent research and industry engagements in the US, EU, and Australia indicate that this market, valued at \$7.6 billion \cite{PayloadSpace2025}, provides a realistic experimental setting for studying complex task allocation problems. These satellite technologies play a critical role in various applications, including agriculture, urban planning, disaster monitoring, and environmental protection \cite{cochran2020earth,elliott2020earth}. This market's unique characteristics and demands provide a valuable opportunity to address the challenges of adaptive task allocation in customer-led, two-sided markets, driving innovation in both theory and practice.

This paper proposes a \textit{customer-led Stackelberg game model} for cost-efficient task allocation in two-sided markets. We introduce a group-buying strategy for customers as leaders with similar needs, allowing them to collaborate through social networks to pool their tasks to reduce costs for individual customers. On the company side, satellite providers can form teams to participate as followers, expanding their capacity to meet increased demand once customer payments are confirmed. we provide theoretical guarantees, demonstrating the existence of an SE for the leader game and an NE for the follower game. Also, we present the uniqueness of NE and SE when utility functions are satisfied with special conditions. To validate the effectiveness of our model, we provide a real-world market to fit the proposed model based on satellite constellation services, which demonstrated that the customer-led Stackelberg game model for task allocation could significantly reduce customer $23\%$ payments, increase company $14$ times revenue, and achieve $100\%$ completed the tasks.

\section{Related Work}

Two-sided markets have been extensively studied, with significant attention given to the interactions between companies and customers and how these interactions shape market dynamics \cite{rysman2009economics}. These studies provide a foundational understanding of platform operations by balancing demand and supply between two interdependent groups, focusing on pricing strategies and the role of externalities \cite{rochet2006two}. In such markets, optimising interactions and resource allocation between the two sides is crucial to maintaining efficiency and maximising overall value.

In this context, game theory offers a robust framework for addressing task allocation problems, particularly in MAS, where the dynamics between agents, such as companies and customers, closely resemble the interactions in two-sided markets. Game theory provides analytical and design tools to account for both individual and collective interests, manage competition \cite{xu2016signaling} and cooperation \cite{chalkiadakis2022computational}, and design incentive mechanisms \cite{zhao2022mechanism,ge2024incentives}. These capabilities enable the identification of optimal task allocation strategies in complex multi-agent environments \cite{picard2021autonomous}, which are particularly relevant in customer-led two-sided markets like the satellite constellation services market.

Among game-theoretic approaches, the Stackelberg game model stands out as a widely applied tool, particularly in scenarios where there is a hierarchical decision-making structure. Traditionally, the Stackelberg game has been applied in company-led markets, where companies act as leaders and customers as followers. For instance, in wireless network markets \cite{rose2014pricing}, companies set fixed prices, and customers make decisions accordingly. This company-led model has been extended to dynamic pricing in industries such as mobile edge computing \cite{CHEN2020273}, where companies allocate cloud resources, and the energy market \cite{zhang2022research}, where operators lead pricing decisions and energy systems follow. The versatility of the Stackelberg game model is further demonstrated in big data marketplaces \cite{li2023optimal} and online markets \cite{ma2021reputation}, highlighting its adaptability to various market structures.

The existing company-led Stackelberg game models inadequately address emerging markets driven by customer demand, particularly when customers increase, personalised service needs become increasingly critical, and customer behaviour becomes more complex and diverse. In such markets, the demand for customer-led Stackelberg game models is growing. In these models, customers act as leaders by proposing tasks, and companies adjust their strategies accordingly. However, task and resource allocation complexity in customer-led environments complicates the modelling process, especially in proving the existence of NE and SE. Furthermore, analysing the mutual influence of customer and company decision-making behaviours becomes increasingly challenging.

\section{Customer-led Two-sided Market}
This section describes the mathematical notation and formulation for a set of customers $X=\{x_1, x_2, \dots \}$ and a set of companies $Y=\{y_1,y_2,\dots\}$, and $|X|\gg |Y|$. A market can facilitate the trading of services, denoted as $S=\{s_1, s_2, \dots, s_l\}$. Each customer needs a set of services, and each service is associated with an independent price they are willing to pay based on their preference. Each company offers a set of services to meet these requirements. Customers can collaborate with the group-buying strategy by proposing tasks based on similar needs to get lower prices through market-pre-defined discount factors. The proposed tasks will be fulfilled by task allocation through teams formed by companies. We first introduce the formal representation of the similarity of services in section 3.1. Then, we formally define the customers and task formation in section 3.2. Lastly, we indicate the representation of team formation for companies and task allocation in section 3.3.  

\subsection{Commonalities in Customers/Companies}
First, customers need to generate tasks with other users who have high similarity in their needs. Conversely, companies need to team up with those with low similarity to increase revenue potentially by extending the team service. An undirected or weighted graph \cite{jiang2022batch} is a common way to represent social networks. Still, it cannot represent the similarity between agents in our model because two vertices in the graph can be different.

We modify social networks into a weighted directed graph (asymmetric edges) to describe the similarity between two agents in our market. A social network $SN=(V, E)$ is a weighted directed graph where $V$ is the set of agents, and $E$ is a set of edges where each edge $e=(v, v')$ represents the social connections between a pair of agents. Let $S_v\subseteq S$ represent the services of agent $v$ needs. For any $(v, v')\in E$, the weight $w_{(v, v')}$ of the edge (or called social distance) can be calculated as follows:
\begin{center}$w_{(v, v')} =\frac{|S_v\bigcap S_{v'}|}{|S_{v'}|}\in [0,1]$ \end{center}
The weight of each edge of the social network represents the match of services between agents. Suppose two agents $v_1$ and $v_2$, and $S_{v_1}=\{s_1,s_2,s_3\}$ and $S_{v_2}=\{s_1\}$. From the $v_1$ point of view, the similarity between $v_1$ and $v_2$ is $\frac{1}{3}$, but the similarity from the $v_2$ point of view is $1$. Any company that can meet the requirement of $v_1$ must meet the requirement of $v_2$, but a company that can meet the requirement of $v_2$ can not guarantee it can meet the requirement of $v_1$. 

\subsection{Customers for Collaborative Tasks}
Each customer $x\in X$ is defined as a tuple $(S_x, p_x)$, where $S_x\subseteq S$ is a set of services that customer $x$ needs, and $p_x: S_x \to \mathbb R$ indicate customer $x$'s payment for each service $s\in S$, where each $p_x(s)$ is a cost that uses service $s$ as input and output the payment of service $s$ that customer $x$ willing to pay after receiving service $s$. Note that the price $p_x(s)$ is determined by customer $x$ based on preference. A customer social network is a weighted directed graph containing all the customers participating in the market as nodes, where weighted edges represent the similarity between customers' needs. We define the customer social network as $CSN=(X, EX)$, where $X$ is a set of customers as nodes and $EX$ is a set of directed edges. $w(x, x')$ is the weight of the directed edge $(x,x')\in EX$ represents the similarity of customer needs between $x$ and $x'$. 

Customers can independently or collaborate with others to generate tasks based on their needs. A task $k\in K$ is defined by a tuple $(X_k, R_k, (d^k_s)_{s\in R_k}, (Pay^k_s)_{s\in R_k})$, where $X_k\subseteq X$ is a set of customers involved in the task $k$, called task members; $R_k=\bigcup_{x\in X_k} S_x$ is a set of services that required to finish task $k$; $d^k_s=\sum_{x\in X_k,s.t. s\in S_x} 1$ indicate that how many customers in the task need service $s$; and $Pay^k_s$ is the payment of each service in task $k$. We apply the group buying strategy to customers whose customers have similar needs and can jointly propose tasks. To encourage customers to generate tasks jointly, we introduce a discount factor $\delta(\cdot)\in (0,1]$, which may vary based on task characteristics or customer types and is pre-determined by the market. This is modelled as a strictly monotonically decreasing function, capturing how payments reduce if more customers $(d^k_s>1)$ need a service in a task. For instance, if only one customer requires a service, $\delta(1)=1$, but $\delta(d^k_s)$ decreases when more customers need that service, with the exact reduction varying depending on the specific task or customer profile. Therefore, the total payment for service $s\in R_k$ task $k$ can be calculated as:
\begin{center}$Pay^k_s=\sum_{x\in X_k,s.t. s\in S_x}\delta(d^k_s)\cdot p_x(s)$\end{center}
Let $K=\{k_1, k_2, \dots, k_l\}$ denote a set of tasks initiated by customers $X$. A set of tasks $K$ is feasible if and only if $\bigcup_{k\in K} X_k=X$ (all customers must be in one of the tasks)  and $X_k\cap X_{k'}=\emptyset$ (one customer can only join one task) for any $k,k'\in K$. Let $\mathcal K$ denote all possible feasible sets of tasks for customers $X$.

\subsection{Companies for Team Formation}
Like customers, companies understand their capacity to provide services, the cost for each service, and the operating costs depending on the company's number of services it can provide. Each company $y\in Y$ is defined by a tuple $(S_y, c_y, o_y)$, where $S_y\subseteq S$ is a set of services that the company provides; $c_y: S_y\to \mathbb R$ indicate the company's cost function, and each $c_y(s)$ denote the cost of the service that the company is making to provide the service; $o_y: S_y\to \mathbb R$ indicate the company's offer function, and each $o_y(s)$ denote the offer of the service that company is making to provide the service and $o_y(s)\geq c_y(s)$ for all $s\in S$. Note that, we assume that $p_x(s)\geq o_y(s)$ for any $x\in X$ and $y\in Y$. Additionally, We assume that collectively, companies can cover all services, but no single company can only offer all services due to large types of services, i.e. $|S_y|<<|S|$, $\bigcup_{y\in Y}S_y=S$ and $\not\exists y\in Y$, such that $S_y=S$. Like the customer social network, a company social network $PSN=(Y, EY)$ is established, where $Y$ represents the set of companies as nodes and $EY$ comprises a set of directed edges. The weight $w(y, y')$ of the directed edge $(y, y')\in EY$ represents the similarity of offered services between companies $y$ and $y'$.

We allow companies to form teams to expand the scope of their services freely. A team $m$ is defined as a tuple $(Y_m, R_m, (q^m_s)_{s_\in R_m})$, where $Y_m$ is a set of companies involved in the team with each company $y\in Y_m$ called a partner; $R_m=\bigcup_{y\in Y_m} S_y$ is a set of services that team $m$ can provide; $q^m_s=\sum_{y\in Y_m,s.t. s\in S_y} 1$ represents the number of partners in the team $m$ that provide service $s$ offered from team $m$. For feasible tasks $K$, let $M_K=\{m_1,m_2,\dots\}$ be a team formation corresponding to $K$. Similar to tasks, a team formation $M_K$ is feasible if and only if $\bigcup_{m\in M_K} Y_m=Y$ and $Y_m\cap Y_{m'}=\emptyset$ for any $m,m'\in M_K$. Let $\mathcal M_K$ denote all possible feasible team formations of the company social network $PSN$. In the rest of this paper, we refer to the tuple $(X, CSN, \mathcal K, Y, PSN, (\mathcal M_K)_{K\in \mathcal K})$ as an instance of the market. 

\begin{example}
\begin{figure}[ht]
    \centering
    \includegraphics[scale=0.1]{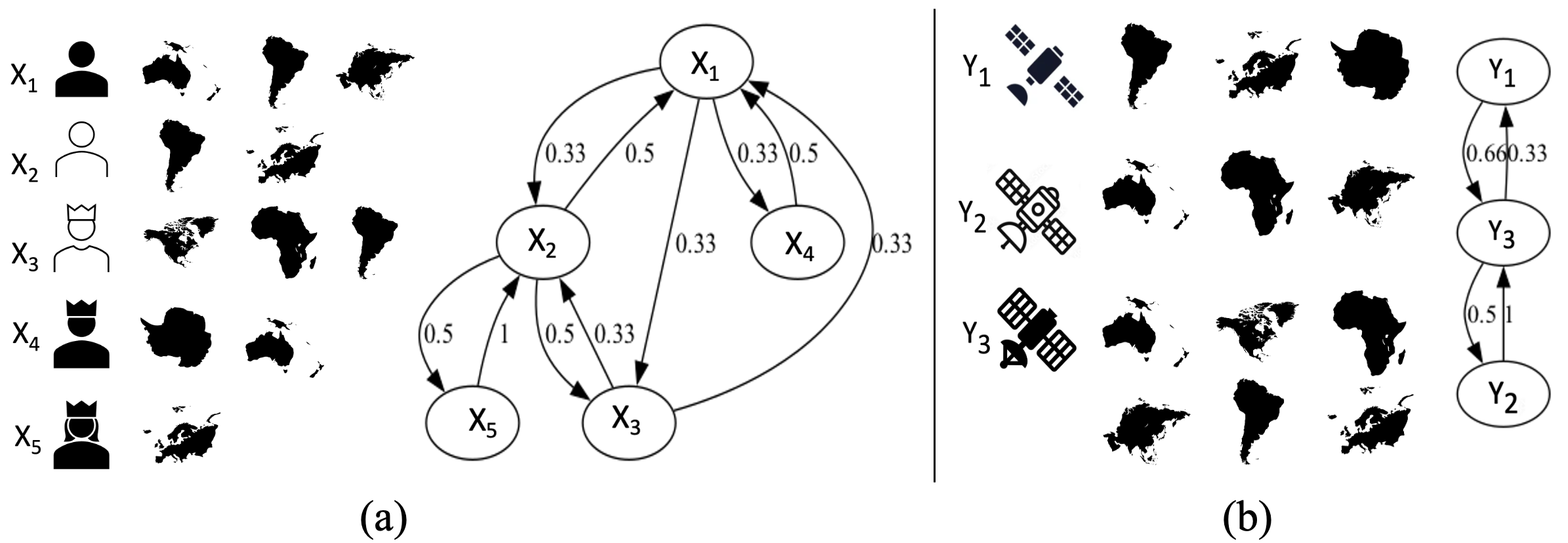}
    \caption{An example of customers and companies: (a) Five customers along with their corresponding customer social network; (b) Three companies along with their corresponding company social network.}\label{PSN_CSN}
\end{figure}

\textbf{Fig.~\ref{PSN_CSN}(a)} shows five customers' requirements and their CSN. Customer $x_1$ needs EO services in Oceania, Africa, and Asia; $x_5$ needs services only in Europe. Weighted directed edges represent the similarity of service needs among customers---for example, $w(x_2, x_5) = 0.5$ indicates 50\% similarity, while $w(x_5, x_2) = 1$ means a complete overlap in needs. \textbf{Fig.~\ref{PSN_CSN}(b)} illustrates the company's PSN and each company's service coverage. Company $y_1$ provides services in Africa, Europe, and Antarctica. Companies can partner to extend service coverage and attract more customer tasks---for instance, $y_1$ partnering with $y_2$ to offer a more comprehensive service package.
\begin{figure}[ht]
    \centering
    \includegraphics[scale=0.11]{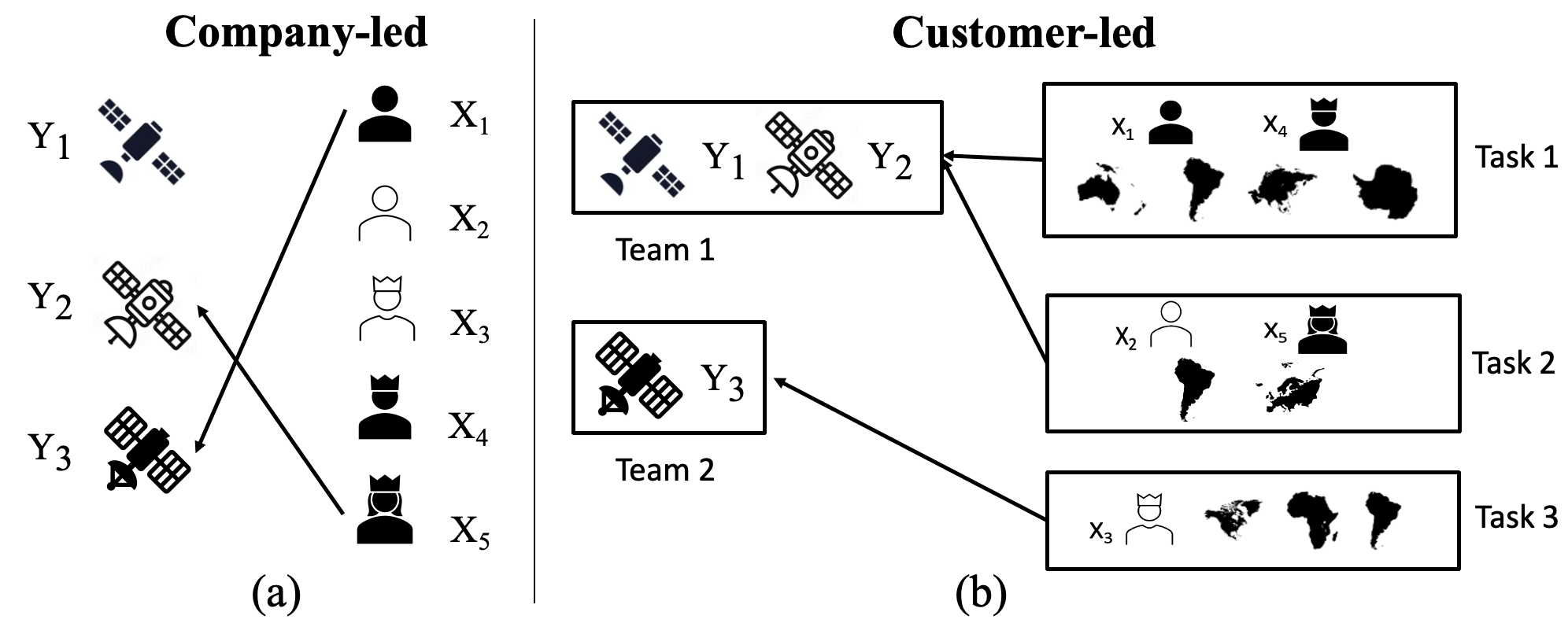}
    \caption{Two examples of task allocation: (a) ``Company-led" indicates Company-led Stackelberg game model companies make price first and customers follow the price; (b) ``Customer-led" refers to our Stackelberg game model facilitating task generation among customers and team formation among companies.}\label{TA_example}
\end{figure}
In the Company-led scenario (\textbf{Fig.~\ref{TA_example}(a)}), customer $x_4$ can't complete their task due to no single company providing all required services, and company $y_1$ secures no tasks because of limited coverage. The Customer-led scenario (\textbf{Fig.~\ref{TA_example}(b)}), where Task Generation and Team Formation are implemented, shows improved outcomes. Customers $x_1$ and $x_4$ (needing services in Oceania) and $x_2$ and $x_5$ (needing services in Europe) benefit from companies $y_1$ and $y_2$ collaborating. 
\end{example}

\section{Customer-led Stackelberg Game Model}
In our market, customers are focused on minimising their costs, whether they undertake a task independently or collaborate with others. When the set of tasks is predetermined, companies seek to maximise their revenue by forming teams or working individually (where a team consists solely of themselves). We model customer and company interactions using the Stackelberg game \cite{stackelberg1952theory}. The leaders (customers) generate a set of feasible tasks for the market, and the followers (companies) form teams to optimise their revenue based on the tasks established by the leaders.

\subsection{Follower Game for Companies}\label{F_G}
The follower game involves companies, denoted as $Y$, acting as players. The set of tasks, $K$, are treated as fixed parameters. The strategy for each company $y$ includes the option to either join a team or work independently. Formally, let $\mathcal A_y$ denote a set of strategies (strategy space) for each company $y$, and each strategy $A_y\subseteq Y$ and $y\in A_y$, which is a subset of companies $Y$ that company $y$ aims to team with. The strategy profile of the follower game $spf=(A_{y_1},\dots, A_{y_{|Y|}})\in SPF$ includes exactly one strategy from the company's strategy space of the follower game, where $SPF=\mathcal A_{y_1}\times \dots \times \mathcal A_{y_{|Y|}}$ denote as all possible strategy profiles. Note that companies' strategies are made simultaneously. Once the strategy profile $spf$ is known, the team formation $M^{spf}_K$ will be determined. The formal definition is as follows.
\begin{definition}
Given a set of tasks $K$ and a strategy profile $spf \in SPF$, the associated team formation of $spf$ is $M^{spf}_K=\{m|\forall y,y'\in Y_m, A_y=A_{y'}\}$.
\end{definition}
To understand the definition above, it can be summarised that all companies have the same strategy on the same team. A simple example for two companies $Y=\{y, y'\}$: if $A_y=\{y\}$ and $A_{y'}=\{y, y'\}$, then the team formation contains two teams $m=\{m_1,m_2\}$, where $Y_{m_1}=\{y\}$ and $Y_{m_2}=\{y'\}$; if $A_y=A_{y'}=\{y,y'\}$ , then the team formation only contains one team $m'=\{m_1\}$, where $Y_{m_1}=\{y,y'\}$, respectively. Note that different strategy profiles can produce the same team formation, such as $spf_1=\{\{y\},\{y'\}\}$ and $spf_2=\{\{y\},\{y, y'\}\}$ have the same team formation $m=\{m_1,m_2\}$, where $Y_{m_1}=\{y\}$ and $Y_{m_2}=\{y'\}$, for two companies $Y=\{y, y'\}$. Let $\mathcal M_K=\{M^{spf}_K:\forall spf\in SPF\}$ represent all possible team formations in the follower game by giving the tasks $K$. We provide more detailed examples of a strategy profile with associated team formation in \textbf{Appendix A.1}.

Next, we introduce the calculation of the service offer price for a team $m$ for task allocation. We assume that the offer price of a service $s \in R_m$ is an independently and identically distributed (i.i.d.) random variable $P^m_s \in [\min_{y \in Y_m} o_y(s), \max_{y \in Y_m} o_y(s)]$. If a company $y \in Y_m$ does not provide service $s$ (i.e., $s \notin R_y$), we define $o_y(s) = \infty$ to exclude it from the competition for that service. Therefore, the offer price for a task $k$ is the sum of multiple i.i.d. random variables:
\begin{center}$P^m_k=\sum_{s\in R_k} P^m_s$\end{center}
It is important to note that the sum of multiple i.i.d. random variables is not an i.i.d. random variable. For example, if two variables are uniformly distributed over $[10, 20]$, their sum follows a triangular distribution over $[10, 40]$. Nevertheless, according to the central limit theorem, when a task involves a large number of services, the sum $P^m_k$ can be approximated by a normal distribution. Using extreme value theory \cite{haan2006extreme}, we can calculate the probability that team \(m\) wins the task, given tasks $K$ and team formation $M_K$, as:
\begin{center}$Pr(m,k)=Pr(P^m_k=\min\{P^{m'}_k: R_k\subseteq R_m\})$\end{center}
This probability represents the likelihood that the offer price from team $m$ is the minimum among all teams with the necessary resources to perform task $k$. The expected profit of team $m$, given a strategy profile $spf$ and its associated team formation $M^{spf}_K$, is:
\begin{center}$E_m[spf]=\sum\limits_{k\in K,s.t. R_k\subseteq R_m} (Pr(m,k)\cdot \sum\limits_{s\in R_k} Pay^k_s)$\end{center}
Here, $Pay_s^k$ represents the payment that team $m$ receives for providing service $s$ as part of task $k$. The total expected profit is the sum of the expected payments for all tasks $k$ that team $m$ has the resources to perform. Finally, the company's revenue, given a strategy profile $spf$, is:
\begin{center}$r_y(spf)=\lambda\cdot\frac{E_m[spf]}{|Y_m|}-\sum\limits_{s\in S_y}c_y(s)$\end{center}
where $\lambda \in (0,1]$ is a parameter that represents how much profit needs to be converted to the cost of team formation. For example, if $\lambda=0.9$, $10\%$ of profit from team formation needs to be reduced as the cost of team formation. Each company's revenue in the team is equal to the team's expected revenue divided by the number of team members minus its own cost. For the follower game, each company aims to maximise its revenue. We can now define the NE for the follower game:

\begin{definition}Given a market instance $(X, CSN, \mathcal{K}, Y, \\PSN, (\mathcal{M}_K)_{K\in \mathcal{K}})$ and a set of feasible tasks $K$, a strategy profile $spf^* = \{A^*_y: \forall y \in Y\}$ is at NE if and only if, for every company $y \in Y$, we have:
\begin{center}$r^K_y(spf^*) \geq r^K_y(spf)$\end{center}
for any alternative strategy profile $spf = \{A_y, A^*_{-y}\} \in SPF$, where $A_y$ represents a new strategy for company $y$ and $A^*_{-y}$ represents the equilibrium strategies of all other companies.
\end{definition}

\subsection{Leader Game for Customers}\label{leader}
The players in the leader game are customers $X$. Each customer, $x \in X$, can generate a task individually or collaborate on a new task with other customers. Formally, let $\mathcal A_x$ denote a set of strategies (strategy space) for each customer $x$, and each strategy $A_x\subseteq X$ and $x\in A_x$, which is a subset of customers $X$ that customer $x$ aims to team with. The strategy profile of the leader game $spl=(A_{x_1},\dots, A_{x_{|X|}})\in SPL$ includes exactly one strategy from the customer’s strategy space of the leader game, where $SPL=\mathcal A_{x_1}\times \dots \times \mathcal A_{x_{|X|}}$ denote as all possible strategy profiles. Note that customers' strategies are made simultaneously. Once the strategy profile $spl$ is known, the task $K^{spl}$ will be identified. The formal definition is as follows.

\begin{definition}
Given a strategy profile of the leader game $spl \in SPL$, the associated tasks are $K^{spl}=\{k|\forall x,x'\in X_k, A_x=A_{x'}\}$.
\end{definition}

To understand the definition above, we can use a similar example in Definition 2 below by replacing company $y$ with customer $x$ to explain the condition of the tasks. We also provide more complex examples in \textbf{Appendix A.2}. For each feasible tasks $K^{spl}$ associated with leader game strategy profile $spl$, the total payment for a customer $x \in X_k$ to complete service if they join task $k\in K^{spl}$ is given by the following equation:
\begin{center}$f^k_x=\sum_{s\in S_x} \delta(d^k_s)\cdot p_x(s)$\end{center} 
where $\delta(d^k_s)$ represents the discount factor and $p_x(s)$ is the initial price. Then, the utility of customer $x\in X$ by given strategy profile of follower and leader game $(spl,spf)$ is 
\begin{equation}\label{utility_leader}
u_x(spl,spf)= \begin{cases}
f^k_x,&x\in X_k \quad\&\quad TA_{M^{spf}_K}(k)\neq\emptyset \\
\xi, & x\in X_k \quad\&\quad TA_{M^{spf}_K}(k)=\emptyset
\end{cases}
\end{equation}
where $f^k_x$ is the utility that denotes the payment for the complete task, and  $\xi>>f^k_x$ represents the utility that the task can not be completed. The Stackelberg game model assumes that the leader acts first, and the followers observe the leader's action before responding. In the proposed Stackelberg game model for the market in Section 3, we aim to identify the optimal strategies for customers and companies, referred to as SE, which is a concept used to describe the interaction between a leader and followers in decision-making processes within a market. The formal definition is as follows. 

\begin{definition} For an instance of a market $(X, CSN, \mathcal K, Y,\\ PSN, (\mathcal M_K)_{K\in \mathcal K})$, a strategy profile of Stackelberg game $(spl^*,spf^*)$ is at SE if and only if 
\begin{itemize}
    \item For any customer $x\in X$, we have $u_x(spl^*,spf^*)\leq u_x(spl, \\spf^*)$ for any $spl=\{A_x,A^*_{-x}\}\in SPL$;
    \item For any company $y\in Y$, we have $r^{K^{spl}}_y(spf^*)\geq r^{K^{spl}}_y(spf)$ for any $spf=\{A_y,A^*_{-y}\}\in SPF$.
\end{itemize}
\end{definition}

In the Stackelberg game model, after formalising task allocation, no customer can achieve a lower payment by changing decisions, and no company can obtain a higher revenue by switching strategies. We have modified task allocation to the tuple $(spl, spf, K^{spl}, M^{spf}_{K^{spl}}, TA_{M^{spf}_{K^{spl}}})$ for the Stackelberg game model of the market. In the rest of this paper, the strategy profile $(spl, spf)$ will be used to represent task allocation.

\subsection{Theoretical Analysis} 

We show the NE in the follower's game and SE in the Stackelberg game as the main theoretical contributions in this paper.

\begin{tcolorbox}[colframe=black, colback=white, boxsep=0pt, left=2pt, right=2pt, top=2pt, bottom=2pt]
\begin{theorem}\label{NE} 
At least one pure strategy Nash equilibrium $spf^*$ exists in the follower game for any given feasible tasks $K$.
\end{theorem}
\end{tcolorbox}
\textbf{See detailed proof of Theorem \ref{NE} in Appendix B.1}. We also provide a corollary to show the uniqueness of NE in the follower game for special conditions.

\begin{tcolorbox}[colframe=black, colback=white, boxsep=0pt, left=2pt, right=2pt, top=2pt, bottom=2pt]
\begin{theorem}\label{SE}
At least one Stackelberg Equilibrium for our Customer-led Stackelberg Game Model.
\end{theorem}
\end{tcolorbox}
\textbf{See detailed proof of Theorem \ref{SE} in Appendix B.2}. We also provide a corollary to show the uniqueness of SE in the Stackelberg Game for special conditions. In addition to the main Theorem above about the existence of NE and SE, our customer-led Stackelberg game model also exhibits some special properties for customers and companies, as demonstrated in the following propositions.

\begin{proposition}\label{Th2} Given feasible tasks $K$ and NE strategy profile $spf^*$, for a team $m$, if a company $y\in m$, and there exist any tasks $k$ such that $TA_{M^{spf^*}_K}(k)=m$ and $\exists s\in R_k\bigcap S_y$, then $y'\not\in Y_m$, where $y\in Y$ and $w(y',y)=1$.
\end{proposition}

\begin{proposition}\label{le1} Given any task allocation $(spl, spf)$, if\\ $u_x(spl, spf)=0$, then $R_k\not\subseteq S_y$ for all $y\in Y$ and $k\in K^{spl}$.
\end{proposition}

\begin{proposition}\label{th6} Given a strategy profile of leader game $spl$ with associated tasks $K^{spl}$, if there exists $(spl,spf)$ and $(spl,spf')$, then $\sum_{y\in Y} r^{K^{spl}}_y(spf)=\sum_{y\in Y} r^{K^{spl}}_y(spf')$. In addition, if two task allocations $(spl,spf)$ and $(spl', spf')$ are both at SE, then $u_x(spl,spf)=u_x(spl', spf')$.
\end{proposition}
\textbf{See detailed proof in Appendix B.3}.

\section{Experiment and Results}
\begin{figure}[ht]
    \centering
    \includegraphics[scale=0.22]{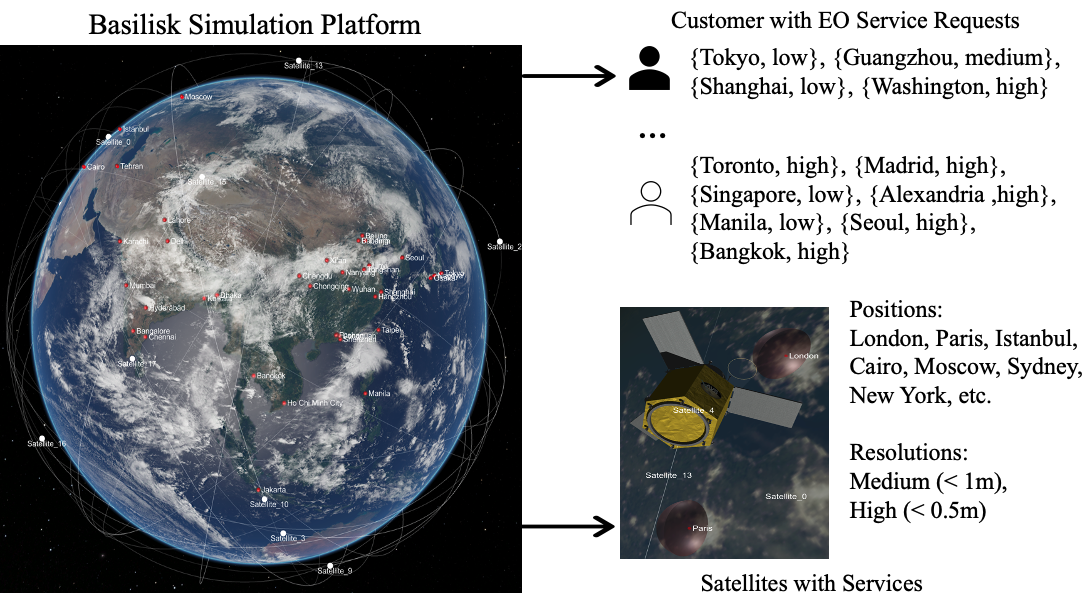}
    \caption{An example of Basilisk simulation with simplified customers' information and company services.}\label{Simulation}
\end{figure}

\subsection{Experiments Settings}
This section presents the experimental setup and results for the satellite constellation services market. The market leverages space technologies to observe and monitor the Earth's surface and atmosphere, primarily through satellite-based systems. We consider a satellite constellation services market as a MAS with customers and satellite companies. Customers request EO images of 80 cities, each with a population over $5$ million, requiring varying resolutions. Satellite companies offer distinct services based on their satellite capabilities. For the testing environment, we extended the open-source Basilisk simulation platform \cite{kenneally2020basilisk}, a modular framework designed for spacecraft simulation. Using Basilisk, we generated LEO satellites with specific orbital parameters—right ascension, argument of perigee, and inclination—derived from city coordinates to determine satellite coverage and ensure realistic service availability. The simulation accounts for real-world constraints such as satellite visibility windows and orbital dynamics. Additionally, customer preferences for image resolution and delivery time are incorporated into the scenario to ensure practical relevance. \textbf{Fig.~\ref{Simulation}} illustrates the simulation setup, showing customers, companies, EO service locations, and image resolutions. Details of the experiment settings and algorithms are provided in \textbf{Appendix C and D}.

\begin{table}[ht]\small
    \centering
    \begin{tabular}{lllll}
        \hline
        \textbf{Scenarios}& \textbf{Customers} & \textbf{Needs} & \textbf{Companies} & \textbf{Services}\\[2pt]
        \hline
        1 & 5000 & [1, 10] & 30 & [10, 30] \\[2pt]
        2 & 5000 & [1, 10] & 15 & [10, 30] \\[2pt]
        3 & 5000 & [1, 10] & 15 & [30, 50] \\[2pt]
        4 & 5000 & [1, 5] & 30 & [10, 30] \\[2pt]
        5 & 10000 & [1, 10] & 30 & [10, 30] \\[2pt]
        \hline
    \end{tabular}
    \caption{The configuration of customers and companies.}
    \label{EXPSETTING}
\end{table}

\textbf{Table \ref{EXPSETTING}} summarises the testing scenarios, varying customer numbers and company service offerings. Scenario 1 is the baseline, while scenario 2 reduces the number of companies. Scenario 3 assumes each company operates two satellites, scenario 4 reduces customer needs, and scenario 5 increases customer numbers. Customers are grouped into five clusters based on task quantities: $200, 400, 600, 800$, and $1000$. Team formation among companies depends on service overlap, with a similarity of $0$ indicating no overlap and $0.5$ allowing up to $50\%$ overlap. Each scenario was repeated 10 times to ensure reliability, and results were averaged under identical parameters starting with scenario 1.

\begin{figure}[htbp]
    \centering
    \includegraphics[width=0.45\textwidth]{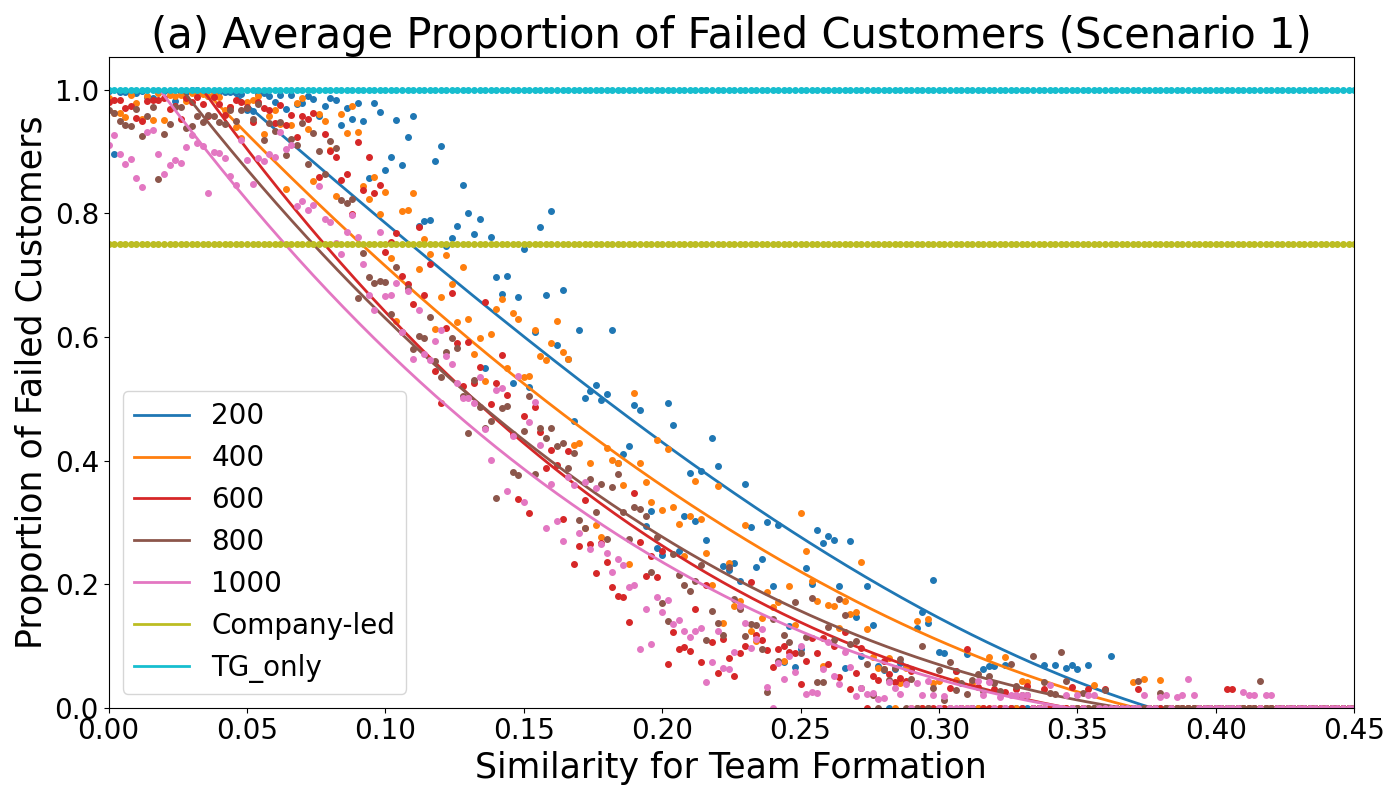}
    \vspace{1em}
    \includegraphics[width=0.45\textwidth]{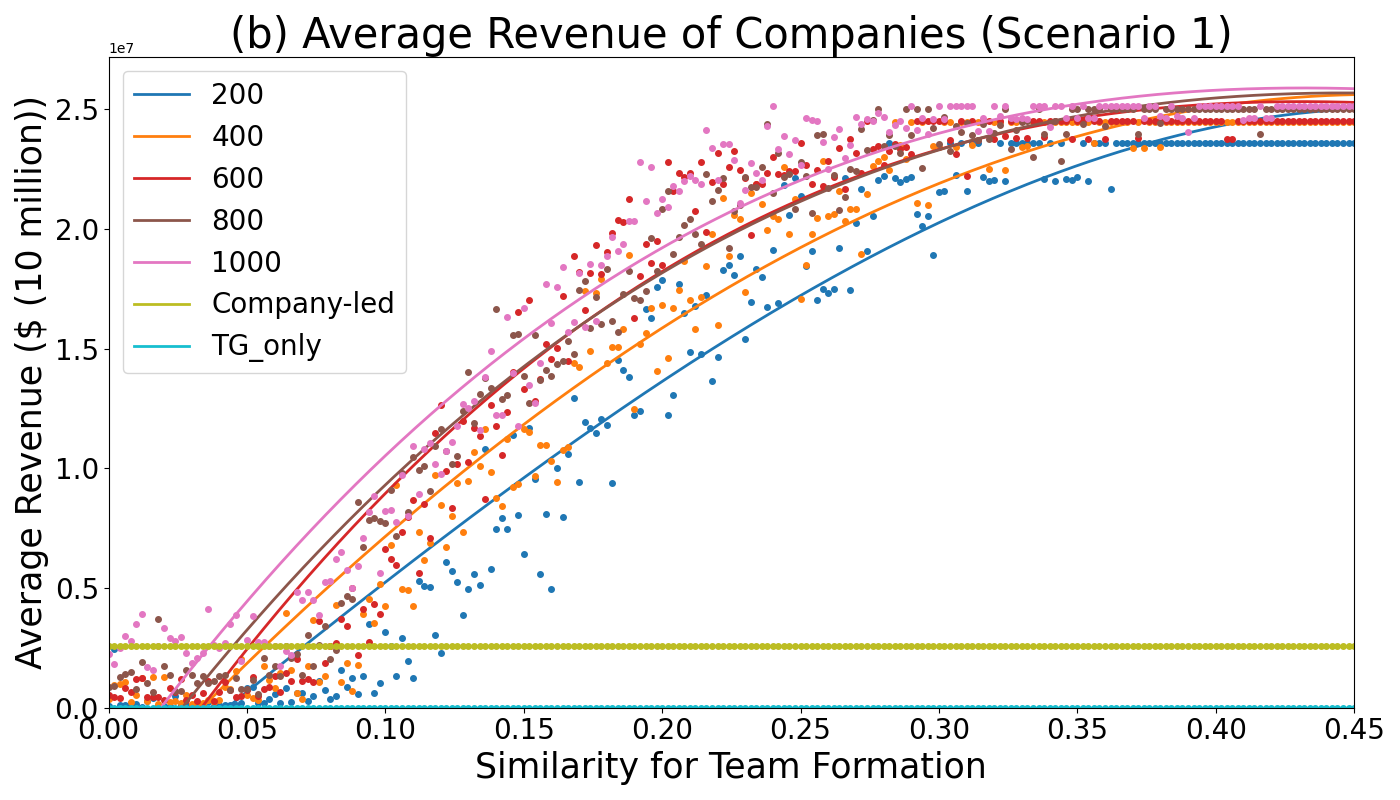}
    \vspace{-5mm}
    \caption{Task allocation performance for satellite constellation services trading in scenario 1: (a) Average proportion of failed customers, and (b) Average revenue of companies. The graphs show original data points and fitting lines.}
    \label{failed_customers_s1}
\end{figure}

\subsection{Numerical Results}
\textbf{Fig.~\ref{failed_customers_s1}(a)} displays the percentage of customers with uncompleted tasks. The yellow line, labelled "Company-led", indicates that approximately $75\%$ of customers have unmet needs in scenarios where tasks are independently submitted without team formation. This line serves as a standard reference. The coloured data points show the distribution of customer task completion across different task quantities and increasing company similarities. Task completion rates are observed to surpass the baseline when the similarity exceeds $0.13$ and approach a $100\%$ success rate as similarity increases. However, setups with fewer tasks exhibit a higher percentage of uncompleted tasks at the same similarity level. For instance, at a similarity of $0.15$, about $60\%$ of customers with 200 tasks remain uncompleted, compared to $40\%$ with 1000 tasks. \textbf{Fig.~\ref{failed_customers_s1}(b)} presents polynomial fitting functions for companies' average revenue relative to task quantities, marked by differently coloured data points. The yellow line acts as a baseline for comparison, while the light blue line represents scenarios where customers generate tasks without company team formation. The graph demonstrates that higher similarity thresholds significantly increase average revenue for companies. Revenue surpasses the baseline once the similarity exceeds $0.08$. Additionally, at a fixed similarity, more tasks correlate with higher average revenue; for example, at $0.15$, revenue from $1000$ tasks is about $1.5$ times that from $200$ tasks.

\begin{figure}[htbp]
    \centering
    \includegraphics[width=0.45\textwidth]{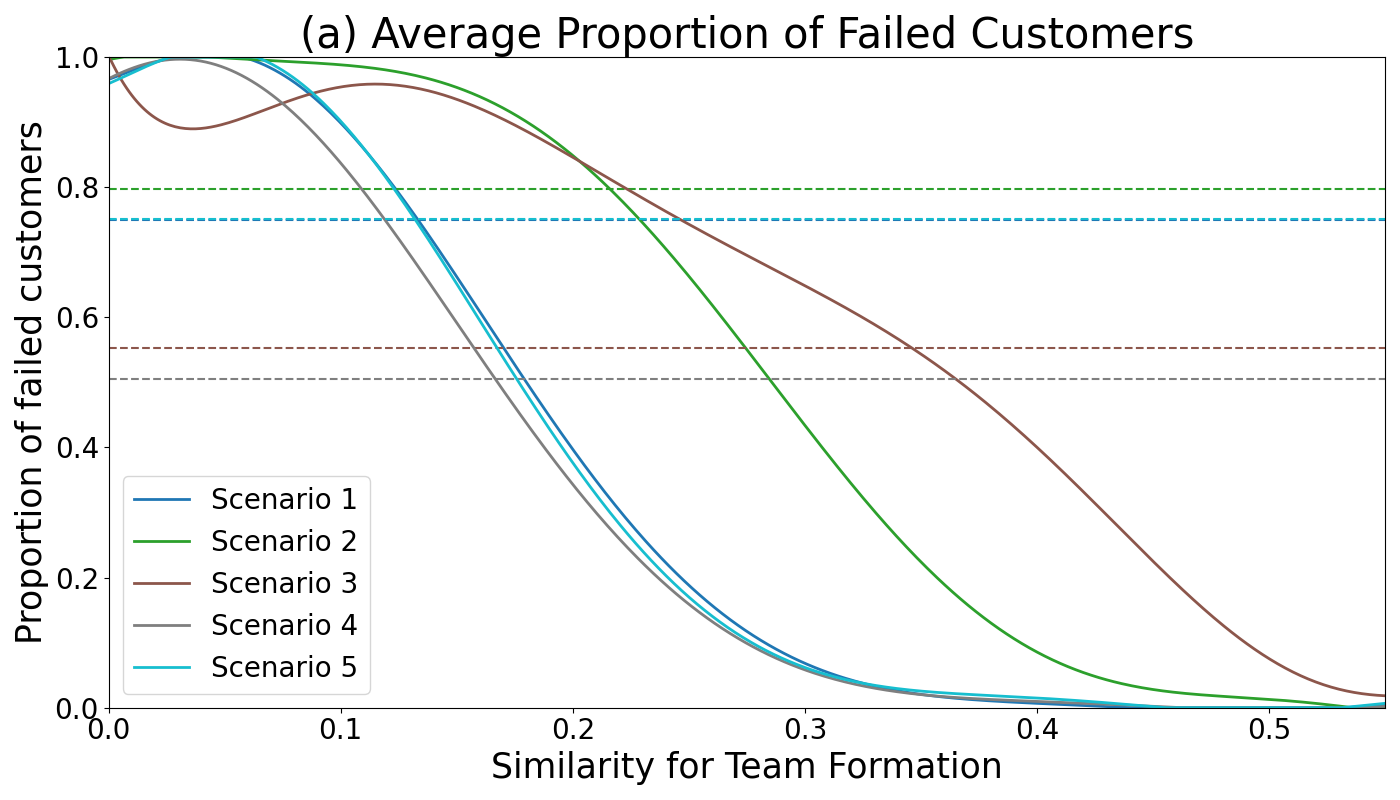}
    \vspace{1em}
    \includegraphics[width=0.45\textwidth]{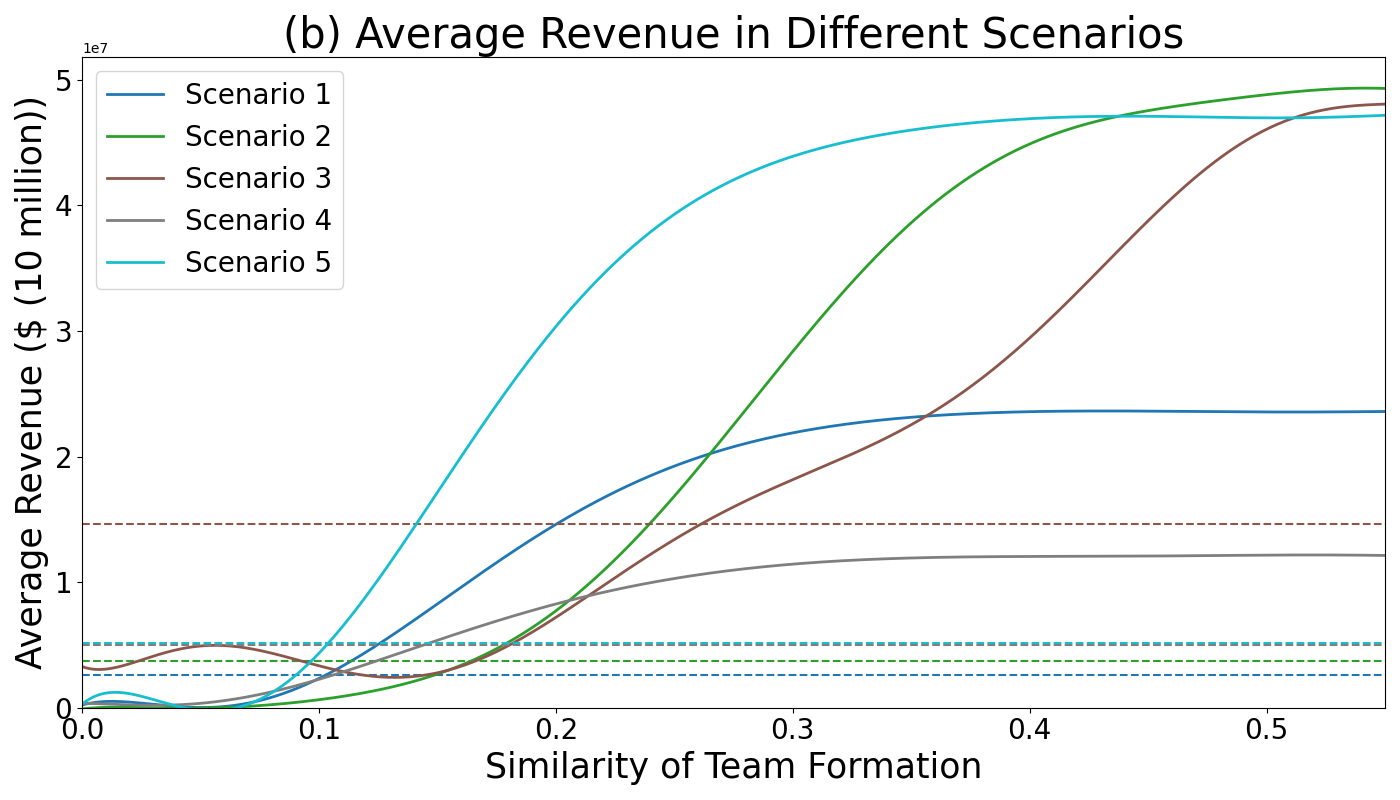}
    \vspace{-5mm}
    \caption{Task allocation performance for satellite constellation services trading of 200 tasks for five scenarios: (a) Average proportion of failed customers, and (b) Average revenue of companies. Solid lines refer to ``Customer-led", and dotted lines refer to ``Company-led" for five scenarios.}
    \label{combined_all}
\end{figure}

\textbf{Fig.~\ref{combined_all}} shows the effect of varying similarity settings on customer allocation failure and average company revenue for $200$ tasks. \textbf{Fig.~\ref{combined_all}(a)} indicates that as similarity increases, scenarios 1, 2, and 5 show nearly identical trends, suggesting similar effects from increasing customer numbers and reducing service demands on customer allocation failures. Reducing the number of companies significantly complicates team formation, especially under Scenario 3, where overlapping service areas hinder team formation at low similarity levels. \textbf{Fig.~\ref{combined_all}(b)} reveals that scenarios 2, 3, and 5 have similar revenue due to customer or company numbers changes, while Scenario 4 has the lowest revenue due to reduced customer demands and payments. \textbf{Table~\ref{NE_all}} presents results of different scenarios in the NE state, comparing 200 tasks to 1000 tasks. These results align with those in \textbf{Table~\ref{NE}}, indicating that task decreases lead to lower company revenues and customer payments, and vice versa. Contrary to other scenarios, Scenario $2$ requires lower similarity to reach NE with $200$ tasks than $1000$ tasks. Full detailed results are in \textbf{Appendix E}.

\begin{table}[h]\footnotesize
    \setlength{\tabcolsep}{0.3pt}
    \begin{tabular}{lllllll}
        \hline
        \textbf{Tasks}&\textbf{SC} & \textbf{Similarity} & \textbf{Utilities} & \textbf{Company-led} & \textbf{Customer-led} & \textbf{PRF}\\[1pt]
        \hline
        200 &1 & 0.364 & Payment & 172,713 & 141,500 &  -18\%    \\[2pt]
           & &       & Revenue  & 2,612,353 & 23,583,488& +900\%  \\[2pt]
        &2 & 0.474 & Payment & 172,681 & 147,413 &  -15\%    \\[2pt]
         &   &       & Revenue  & 3,918,280 & 49,137,997 & +1254\%  \\[2pt]
        &3 & 0.538 & Payment & 173,392 & 148,033  & -14.5\%    \\[2pt]
         &   &       & Revenue & 14,435,611 & 49,344,631 & +341\%  \\[2pt]
        &4 & 0.36 & Payment & 94,196 & 72,760  & -23\%    \\[2pt]
         &   &       & Revenue & 5,060,874 & 12,126,813 &  +239\% \\[2pt]
        &5   & 0.344 & Payment & 172,747 & 141,438  & -19\%    \\[2pt]
         &   &       & Revenue & 5,475,274 & 47,146,183 & +861\%    \\[2pt]
        \hline
        1000 &1   & 0.304 & Payment & 172,713 & 150,816    &  -12.5\%  \\[2pt]
        &    &       & Revenue & 2,578,431 & 25,136,018  &  +974\% \\[2pt]
        &2   & 0.528 & Payment & 173,427 & 153,597 &  -11\%    \\[2pt]
        &    &       & Revenue  & 3,580,109 & 51,199,116 & +1430\%  \\[2pt]
        &3 & 0.334 & Payment & 172,566 & 152,874  & -11\%    \\[2pt]
        &    &       & Revenue & 15,011,309 & 50,958,118 & +339\%  \\[2pt]
        &4 & 0.274 & Payment & 94,422 & 80,815  & -14\%    \\[2pt]
        &    &       & Revenue & 5,038,154 & 13,469,243 &  +267\% \\[2pt]
        &5   & 0.3 & Payment & 172,293 & 149,938  & -13\%    \\[2pt]
        &    &       & Revenue & 5,104,622 & 49,979,608 & +979\%    \\[2pt]
        \hline
    \end{tabular}
    \caption{Average customers' payment and companies' revenue at Nash equilibrium in 200 and 1000 tasks for five scenarios (SC).}
    \label{NE_all}
\end{table}

\section{Conclusion}
In conclusion, this paper introduced a customer-led Stackelberg game model for cost-efficient task allocation in the satellite constellation services market. The model reduces payment costs by enabling customers to generate tasks through group-buying strategies and allows companies to form teams to optimise resources. Theoretical guarantees ensure the existence of both Stackelberg and Nash equilibriums. A case study on the Earth observation market shows significant benefits, including a 23\% reduction in customer payments and a 6.7-time increase in company revenues.

\bibliographystyle{named}
\bibliography{ARXIV_IJCAI2025_Jianglin/Arxiv}

\appendix
\section{Examples of Strategy Profiles for Follower and Leader Game}

\subsection{Follower Game}
Let us extend the example from Section 4.1 of the main paper to consider a scenario involving three companies \( Y =\\\{y_1, y_2, y_3\} \). Each company can work independently or collaborate with other companies to complete tasks. The strategies available to each company can be defined as \\\( A_{y_1}, A_{y_2}, A_{y_3} \), where each strategy represents the set of companies a given company wishes to collaborate with.

\subsubsection*{Example 1: Same Team Formation Resulting from Different Strategy Profiles}
\begin{itemize}
    \item \textbf{Companies' Choices 1:}
        \begin{itemize}
            \item $A_{y_1} = \{y_1\}$ (Company $y_1$ chooses to work independently)
            \item $A_{y_2} = \{y_2, y_3\}$ (Company $y_2$ wants to collaborate with $y_3$)
            \item $A_{y_3} = \{y_3\}$ (Company $y_3$ chooses to work independently)
        \end{itemize}
    \item \textbf{Companies' Choices 2:}
        \begin{itemize}
            \item $A_{y_1} = \{y_1, y_2\}$ (Company $y_1$ chooses to work independently)
            \item $A_{y_2} = \{y_2\}$ (Company $y_2$ wants to collaborate with $y_3$)
            \item $A_{y_3} = \{y_1,y_3\}$ (Company $y_3$ chooses to work independently)
        \end{itemize}
    \item \textbf{Team Formation:} In both strategy profiles 1 and 2, the result is the same: three independent teams are formed 
        \begin{itemize}
            \item Team $m_1$, where $Y_{m_1}=\{y_1\}$;
            \item Team $m_2$, where $Y_{m_2}=\{y_2\}$;
            \item Team $m_3$, where $Y_{m_3}=\{y_3\}$.
        \end{itemize}
\end{itemize}

\subsubsection*{Example 2: Different Strategy Profiles Leading to Two Teams}
\begin{itemize}
    \item \textbf{Companies' Choices:}
        \begin{itemize}
        \item $A_{y_1} = \{y_1, y_2\}$ (Company $y_1$ wants to collaborate with $y_2$)
        \item $A_{y_2} = \{y_1, y_2\}$ (Company $y_2$ wants to collaborate with $y_1$)
        \item $A_{y_3} = \{y_3\}$ (Company $y_3$ chooses to work independently)
        \end{itemize}
    \item \textbf{Team Formation:} In this case, $y_1$ and $y_2$ form a collaborative team, while $y_3$ works independently, resulting in two teams:
        \begin{itemize}
            \item Team $m_1$, where $Y_{m_1}=\{y_1,y_2\}$;
            \item Team $m_3$, where $Y_{m_3}=\{y_3\}$.
        \end{itemize}
\end{itemize} 

\subsubsection*{Example 3: Complex Strategy Profiles for Team Formation}
\begin{itemize}
    \item \textbf{Companies' Choices:}
        \begin{itemize}
            \item $A_{y_1} = \{y_1, y_2, y_3\}$ (Company $y_1$ wants to collaborate with all companies)
            \item $A_{y_2} = \{y_2, y_3\}$ (Company $y_2$ wants to collaborate with $y_3$)
            \item $A_{y_3} = \{y_1, y_3\}$ (Company $y_3$ wants to collaborate with $y_2$)
        \end{itemize}
    \item \textbf{Team Formation:} Since Company 1 wants to collaborate with Companies 2 and 3, Company 2 wants to collaborate with Company 3, but Company 3 wants to collaborate with Company 1. This results in a failed team formation strategy, leading to all companies forming independent teams.
        \begin{itemize}
            \item Team $m_1$, where $Y_{m_1}=\{y_1\}$;
            \item Team $m_2$, where $Y_{m_2}=\{y_2\}$;
            \item Team $m_3$, where $Y_{m_3}=\{y_3\}$.
        \end{itemize}
\end{itemize}

\subsection{Leader Game}
We provided examples of strategy profiles with three customers $X = \{x_1, x_2, x_3\}$ for the leader's game for customers by replacing the companies with customers from the above examples. 

\subsubsection*{Example 4: Complex Strategy Profiles for Task Generation}
\begin{itemize}
    \item \textbf{Customers' Choices:}
        \begin{itemize}
            \item $A_{x_1} = \{x_1, x_2, x_3\}$ (Customer $x_1$ wants to collaborate with all customers)
            \item $A_{x_2} = \{x_2, x_3\}$ (Customer $x_2$ wants to collaborate with $x_3$)
            \item $A_{x_3} = \{x_1, x_3\}$ (Customer $x_3$ wants to collaborate with $x_2$)
        \end{itemize}
    \item \textbf{Task Generation:} Since Customer 1 wants to collaborate with Customer 2 and 3, Customer 2 wants to collaborate with Customer 3, but Customer 3 wants to collaborate with Customer 1. This results in a failed task generation.
        \begin{itemize}
            \item Independent task $k_1$, where $X_{k_1}=\{x_1\}$.
            \item Independent task $k_2$, where $X_{k_2}=\{x_2\}$.
            \item Independent task $k_3$, where $X_{k_2}=\{x_3\}$.
        \end{itemize}
\end{itemize}

\subsubsection*{Example 5: Strategy Profile Leading to One Joint Task}
\begin{itemize}
    \item \textbf{Customers' Choices:}
        \begin{itemize}
            \item $A_{x_1} = \{x_1, x_2, x_3\}$ (Customer $x_1$ wants to collaborate with all customers)
            \item $A_{x_2} = \{x_1, x_2, x_3\}$ (Customer $x_2$ wants to collaborate with all customers)
            \item $A_{x_3} = \{x_1, x_2, x_3\}$ (Customer $x_3$ wants to collaborate with all customers)
        \end{itemize}
    \item \textbf{Task Generation:} Since consumers 1, 2, and 3 all wish to collaborate and propose a task together, the strategy profile includes only one joint task.
        \begin{itemize}
            \item Joint task $k_1$, where $X_{k_1}=\{x_1,x_2,x_3\}$.
        \end{itemize}
\end{itemize}

\section{Detailed Proof for Theorem 1, Theorem 2 and Propositions} 

\subsection{Proof for Theorem 1}
We first need to analyse Nash equilibrium in the follower's game to investigate the properties of our designed marketplace based on the Stackelberg game. 

\begin{tcolorbox}[colframe=black, colback=white, boxsep=0pt, left=2pt, right=2pt, top=2pt, bottom=2pt]
\begin{theorem}
At least one pure strategy Nash equilibrium $spf^*$ exists in the follower game for any given feasible tasks $K$.
\end{theorem}
\end{tcolorbox}

\begin{proof}
We prove the existence of a pure strategy Nash equilibrium (PSNE) by showing that this follower game is an ordinal potential game \cite{voorneveld1997characterization,ewerhart2020ordinal}. Specifically, we construct an ordinal potential function and demonstrate that for any unilateral strategy change by a company, the direction of change in that company’s payoff matches the direction of change in the potential function. Let $spf$ be a strategy profile, and let $M^{spf}_K$ denote the resulting team formation. Define the \emph{potential function} $\Phi(spf)$ as the sum of the expected profits of all teams under $spf$:
\[
\Phi(spf) \;=\; \sum_{m \in M^{spf}_K} E_m[spf],
\]
where
\[
E_m[spf] \;=\; \sum_{\substack{k \in K \\ R_k \subseteq R_m}}
    \Bigl( \Pr(m,k) \cdot \sum_{s \in R_k} \textit{Pay}_s^k \Bigr)
\]
is the expected profit of team $m$ under the strategy profile $spf$. Here, $\Pr(m,k)$ is the probability that team $m$ wins task $k$, and $\textit{Pay}_s^k$ is the payment received for providing service $s$ as part of task $k$. We assume each company’s cost is either fixed or does not affect the direction of any unilateral payoff change, so the relevant part for potential-function analysis is the team’s expected profit.

We need to prove that for any company $y \in Y$ and any unilateral strategy change from $spf$ to $spf'$, the following equivalence holds:
\small{\[
r_y(spf') \;-\; r_y(spf) \;>\; 0
\quad \Longleftrightarrow \quad
\Phi(spf') \;-\; \Phi(spf) \;>\; 0,
\]}
where $r_y(spf)$ denotes the payoff of company $y$ under strategy profile $spf$.
\begin{itemize}
    \item Let $m_y$ be the \emph{old team} of company $y$ under $spf$.  
    \item Let $m'_y$ be the \emph{new team} of company $y$ under $spf'$ (i.e., the team $y$ joins in the new strategy).
\end{itemize}
  
When $y$ unilaterally deviates from $spf$ to $spf'$, it potentially leaves $m_y$ and joins $m'_y$. Other teams not involving $y$ remain unaffected. For convenience (as is common in potential-game proofs), we set $\lambda = 1$ or treat $\lambda > 0$ as a positive constant that does not affect the inequality’s direction. Then, the company $y$’s payoff change simplifies to

\[
\begin{aligned}
\Delta r_y
\;=\; &
r_y(spf') \;-\; r_y(spf) \\
\;=\; &
\biggl(\frac{E_{m'_y}[spf']}{\lvert Y_{m'_y}\rvert}\biggr)
\;-\;
\biggl(\frac{E_{m_y}[spf]}{\lvert Y_{m_y}\rvert}\biggr)
\end{aligned}
\]

Only the old team $m_y$ and the new team $m'_y$ may have different expected profits under $spf'$ versus $spf$. Denote:

- $E_{m_y}[spf]$: the expected profit of the old team $m_y$ (including $y$) under $spf$.  
- $E_{m'_y}[spf']$: the expected profit of the new team $m'_y$ (including $y$) under $spf'$.  
- $E_{m'_y}^{-y}(spf)$: the expected profit of team $m'_y$ \emph{before} $y$ joins (i.e., under $spf$, when $m'_y$ did not include $y$).  
- $E_{m_y}^{+y}(spf')$: the expected profit of team $m_y$ \emph{after} $y$ leaves (i.e., under $spf'$, when $m_y$ no longer has $y$).

Thus, the change in the potential function is:

\[
\begin{aligned}
\Delta \Phi
&\;=\;
\Phi(spf') \;-\; \Phi(spf) \\
&\;=\;
\Bigl(E_{m'_y}[spf'] - E_{m'_y}^{-y}(spf)\Bigr) \\
&\quad+\;
\Bigl(E_{m_y}^{+y}(spf') - E_{m_y}[spf]\Bigr)
\end{aligned}
\]

\textbf{Direction 1:} $\Delta r_y > 0 \;\Longrightarrow\; \Delta \Phi > 0$.

Suppose $\Delta r_y > 0$. Then
\[
\frac{E_{m'_y}[spf']}{\lvert Y_{m'_y}\rvert}
\;>\;
\frac{E_{m_y}[spf]}{\lvert Y_{m_y}\rvert}
\]
Since $y$’s joining $m'_y$ improves its average payoff, it generally indicates that $m'_y$’s expected profit increases when $y$ joins:
\[
E_{m'_y}[spf'] \;>\; E_{m'_y}^{-y}(spf)
\]
Meanwhile, the old team $m_y$ may lose some capability or resource after $y$ leaves, implying
\[
E_{m_y}^{+y}(spf') \;\le\; E_{m_y}[spf]
\]
For $\Delta r_y$ to be strictly positive, the net effect of these changes on total expected profits must also be positive:
\[
\begin{aligned}
\Delta \Phi 
&\;=\; 
\Bigl(E_{m'_y}[spf'] - E_{m'_y}^{-y}(spf)\Bigr) \\
&\quad + 
\Bigl(E_{m_y}^{+y}(spf') - E_{m_y}[spf]\Bigr) \\
&\;>\; 0
\end{aligned}
\]
Hence, $\Delta \Phi > 0$.

\medskip
\noindent
\textbf{Direction 2:} $\Delta \Phi > 0 \;\Longrightarrow\; \Delta r_y > 0$.

Conversely, if $\Delta \Phi > 0$, it means the total expected profit of $m'_y$ and $m_y$ combined increases after $y$’s move from $spf$ to $spf'$. Thus, $m'_y$’s expected profit relative to its size increases sufficiently so that
\[
\frac{E_{m'_y}[spf']}{\lvert Y_{m'_y}\rvert}
\;>\;
\frac{E_{m_y}[spf]}{\lvert Y_{m_y}\rvert}
\]
Therefore, $y$’s own payoff must also increase, i.e., $\Delta r_y > 0$. 

Since each unilateral “profitable” deviation strictly increases $\Phi(spf)$, and the strategy sets are finite, $\Phi(spf)$ must be bounded. Hence, the game possesses the finite improvement property (FIP): any sequence of unilateral profitable deviations terminates in finitely many steps. The terminal outcome, where no player can unilaterally deviate to increase payoff, is by definition a pure strategy Nash equilibrium. Thus, we conclude that this follower game has at least one pure strategy Nash equilibrium.
\end{proof}

Although Theorem~\ref{NE} guarantees the existence of at least one pure strategy Nash equilibrium (PSNE) in the follower game, in a general ordinal potential game, there may be multiple PSNEs. To ensure \emph{uniqueness}, we introduce an additional strict monotonicity assumption on the potential function \(\Phi\), imposing a stronger structural constraint.

\begin{assumption}[Strict Monotonicity]\label{ass:singlepeaked}
On the finite strategy space, the potential function \(\Phi(spf)\) has a unique global maximum. That is, there exists exactly one strategy profile \(spf^*\) such that
\[
\Phi(spf^*) > \Phi(spf)
\quad
\forall\, spf \neq spf^*,
\]
and there is no other \(spf' \neq spf^*\) for which \(\Phi(spf') = \Phi(spf^*)\). 
This property is referred to as strict monotonicity, ruling out ties for the highest potential value.
\end{assumption}

\begin{corollary}[Uniqueness of PSNE]\label{cor:uniqueness}
Under the same conditions of Theorem~\ref{NE}, if the potential function \(\Phi\) further satisfies the Assumption~\ref{ass:singlepeaked}, then there is a \textbf{unique} pure strategy Nash equilibrium in the follower game.
\end{corollary}

\begin{proof}[Proof of Uniqueness]
By Theorem~\ref{NE}, the game admits at least one PSNE. Suppose there are two different pure strategy Nash equilibrium \(spf^1 \neq spf^2\). In an ordinal potential game, each PSNE corresponds to a local maximum of \(\Phi\), meaning that no unilateral deviation can raise the potential any further at either \(spf^1\) or \(spf^2\).

If \(\Phi(spf^1) = \Phi(spf^2)\) yet \(spf^1 \neq spf^2\), then they share the same potential value but are distinct strategy profiles, contradicting Assumption~\ref{ass:singlepeaked}. That assumption rules out two different global maxima having the same \(\Phi\)-value. Hence, \(\Phi\) can only have one global maximiser \(spf^*\).

Because a global maximiser in an ordinal potential game is a pure strategy Nash equilibrium (no single-player deviation can improve the potential), and no other strategy profile can tie with \(\Phi(spf^*)\), the PSNE must be \emph{unique}.
\end{proof}

\subsection{Proof for Theorem 2}
\begin{tcolorbox}[colframe=black, colback=white, boxsep=0pt, left=2pt, right=2pt, top=2pt, bottom=2pt]
\begin{theorem}
At least one Stackelberg Equilibrium for our Customer-led Stackelberg Game Model.
\end{theorem}
\end{tcolorbox}

\begin{proof}
For each $spl \in SPL$, the follower game defined by $K^{spl}$ has at least one pure-strategy Nash equilibrium. Denote by
\[
  \mathrm{BR}(spl)\;\in\;SPF
\]
The chosen pure NE of the companies (followers) is given $spl$. Thus, $\mathrm{BR}(spl)$ is the followers' best response to the leader strategy $spl$. Once $\mathrm{BR}(spl)$ is fixed for any $spl$, each customer $x$'s cost function becomes
\[
  \tilde{u}_x(spl)
  \;=\;
  u_x\bigl(spl,\;\mathrm{BR}(spl)\bigr).
\]
Because each $\mathcal{A}_x$ is finite and $x$'s cost is well-defined for all $spl$, the leaders (customers) face a finite normal-form game in $SPL$ with cost functions $\{\tilde{u}_x\}$. By the classic Nash Existence Theorem for finite games, there is at least one (possibly mixed) Nash equilibrium in the leaders' game.

Suppose there were a mixed strategy $\sigma_x$ for some customer $x$ that assigns positive probability to an action causing the task to fail. Because $\xi \gg \max_k f^k_x$, incurring $\xi$ in any portion of the mixture increases the expected cost above that of a pure strategy guaranteeing success. Hence, such a mixed strategy is strictly dominated by a pure strategy that avoids task failure. According to this argument, no customer will mix over failing actions in equilibrium. Therefore, the leaders' Nash equilibrium can be chosen in pure strategies; denote it by $spl^*$. Let $spf^* = \mathrm{BR}(spl^*)$ be the followers' pure-strategy best response to $spl^*$. Then:

\begin{itemize}
  \item \emph{Follower optimality}: $spf^*$ is a pure NE given $K^{spl^*}$, so no company can deviate to improve its payoff. 
  \item \emph{Leader optimality}: $spl^*$ is a pure NE in the leaders' finite game (with a large penalty), so no customer can deviate from lowering its cost.
\end{itemize}
Hence, $(spl^*, spf^*)$ constitutes a pure-strategy Stackelberg Equilibrium in the customer-led setting.
\end{proof}

\begin{corollary}[Uniqueness of SE under Strict Monotonicity]
\label{cor:unique_speaked}
In the customer-led Stackelberg game described above, suppose we define the \emph{global cost} as
\[
  \Phi(spl) 
  \;=\;
  \sum_{x \in X} \;\tilde{u}_x(spl)
  \;=\;
  \sum_{x \in X} \;u_x\bigl(spl,\,\mathrm{BR}(spl)\bigr).
\]
Assume that $\Phi(\cdot)$ satisfies the strict monotonicity property on the finite set $SPL$, meaning there is a unique global minimiser 
\[
  spl^\dagger \quad \text{such that}\quad 
  \Phi(spl^\dagger) < \Phi(spl)\quad \forall\,spl\neq spl^\dagger,
\]
with \emph{no ties} or other points achieving $\Phi(spl^\dagger)$. Then there is \textbf{exactly one} pure-strategy Stackelberg Equilibrium in this model.
\end{corollary}

\begin{proof}
We already know the game has at least one pure-strategy Stackelberg Equilibrium $(spl^*, spf^*)$ (from Theorem \ref{SE}). Concretely,
\begin{center}$spf^* = \mathrm{BR}(spl^*)$\end{center}
and $spl^*$ is a pure NE among the leaders once $\mathrm{BR}(\cdot)$ is fixed. Suppose, for contradiction, that there is a second distinct pure-strategy SE $(\widehat{spl}, \widehat{spf})$ with $\widehat{spl}\neq spl^*$. By definition, 
\begin{center}$\widehat{spf} = \mathrm{BR}(\widehat{spl})$\end{center}
and $\widehat{spl}$ is also a pure NE in the leaders' finite game. Define $\Phi(spl) = \sum_{x\in X} u_x\bigl(spl,\mathrm{BR}(spl)\bigr)$. By hypothesis, $\Phi(spl)$ has a \emph{unique} global minimiser $spl^\dagger \in SPL$. That is,
\[
  \Phi(spl^\dagger) < \Phi(spl)
  \quad
  \forall\,spl\neq spl^\dagger,
\]
and there are no ties at the same cost level $\Phi(spl^\dagger)$. Since $spl^*$ is a stable solution among the leaders, no leader $x$ can deviate from lowering its cost, implying $spl^*$ must be at least a \emph{local} minimiser with respect to unilateral deviations. However, in a \emph{strict monotonicity} environment, a local minimum that achieves the minimal possible $\Phi$-value must be \emph{the} global minimiser if it is stable. Hence
\[
  \Phi(spl^*) \;=\; \min_{spl\in SPL}\;\Phi(spl),
\]
which, by uniqueness, forces
\[
  spl^* \;=\; spl^\dagger.
\]

Consider $\widehat{spl}\neq spl^*$; if it too is a stable pure NE for the leaders, it would likewise have to achieve the same globally minimal cost $\Phi(\widehat{spl})$, contradicting the uniqueness of $spl^\dagger$. Therefore, $\widehat{spl}$ cannot be equally optimal. As a result, $(\widehat{spl}, \widehat{spf})$ fails to be a stable solution under strict monotonicity. Thus there is no second distinct pure-strategy equilibrium. Consequently, \emph{the} pure-strategy Stackelberg Equilibrium $(spl^*, spf^*)$ we found is indeed \textbf{unique}.
\end{proof}

\subsection{Proof for Propositions}

\begin{proposition} Given feasible tasks $K$ and NE strategy profile $spf^*$, for a team $m$, if a company $y\in m$, and there exist any tasks $k$ such that $TA_{M^{spf^*}_K}(k)=m$ and $\exists s\in R_k\bigcap S_y$, then $y'\not\in Y_m$, where $y\in Y$ and $w(y',y)=1$.
\end{proposition}

\begin{proof} We assume that for a given NE strategy profile $spf^*$ and its $(M^{spf^*}_K, TA_{M^{spf^*}_K})$, there is a team $m\in M_K$ such that $y,y'\in Y_m$ and $w(y',y)=1$, and a task $k$ such that $R_k=\{s\}$ and $s\in S_y$. Suppose a teams $M^{spf}_K$ indicate by $spf$, where $m',m''\in M^{spf}_K$, with $Y_{m'}=\{y'\}$ and $Y_{m''}=Y_m\setminus \{y'\}$. It follows that $o^{m'}_s<o^m_s$, leading to $TA_{M^{spf}_K}(k)=m'$ and $Pay^K_{y'}(M_K,TA_{M_K})>Pay^K_{y'}(M^{spf^*}_K, TA_{M^{spf^*}_K})$. Consequently, $r^K_{y'}(spf)>r^K_{y'}(spf^*)$, which indicates that $spf^*$ is not a Nash equilibrium.
\end{proof}
In the proposition above, we have proven that if two companies offer the same services and customers need them, these companies cannot be on the same team in the Nash equilibrium strategy. In the subsequent theorem, we demonstrate that no individual company can independently service that task if a task cannot be executed in a given task allocation.
\begin{proposition} Given any task allocation $(spl, spf)$, if \\$u_x(spl, spf)=0$, then $R_k\not\subseteq S_y$ for all $y\in Y$ and $k\in K^{spl}$.
\end{proposition}
The proof of proposition \ref{le1} is straightforward. The condition $u_x(spl, spf)=0$ implies that $R_k\not\subseteq R_m$ for any $m\in M^{spf}_K$. Assume a task $k\in K$ and a company $y\in Y$ such as $R_k\subseteq S_y$. Consequently, a team $m$, which includes $y\in Y_m$, must exist, leading to the conclusion that $R_k \subseteq R_m$. This results in $TA_{M^{spf}_{K^spl}}=m$.

\begin{proposition} Given a strategy profile of leader game $spl$ with associated tasks $K^{spl}$, if there exists $(spl,spf)$ and $(spl,spf')$, then $\sum_{y\in Y} r^{K^{spl}}_y(spf)=\sum_{y\in Y} r^{K^{spl}}_y(spf')$. In addition, if two task allocations $(spl,spf)$ and $(spl', spf')$ are both at SE, then $u_x(spl,spf)=u_x(spl', spf')$.
\end{proposition}
Theorems \ref{NE} and \ref{SE} show that a SE exists in our market model; therefore, this theorem can be proven using the Lemma 2.9 of \cite{roughgarden2002bad}. Proposition \ref{th6} proves that all strategy profiles are at the Nash equilibrium state for given tasks. Consequently, the follower's game yields the same social welfare, representing all companies' total utility. Furthermore, if more than one SE exists in the customers' game, each customer's utility remains the same.

\section{The Detailed Setting of Our Experiment}

Based on our research and analysis of public websites and databases, the current EO service market reveals that EO image prices range between $\$8$ and $\$12$ per square kilometre. Consequently, we adopted a standard rate of $\$10$ per square kilometre. Price variations are due to the differing sizes of administrative regions, as exemplified by New York City, where the standard price was calculated at $\$120,930$ based on its area of $12,093 km^2$. EO services are offered at three image resolutions: low ($<1.5m$), medium ($<1m$), and high ($<0.5m$), with costs set at 110\% of the standard price for medium resolution and 120\% for high resolution. Each satellite incurs an EO imaging cost of $\$3$ per square kilometre, with the baseline operational cost for each company set at $\$5000$ multiplied by the number of services ($|S_y|$) they offer. The formation cost is defined as $\lambda(\cdot)\in [0,0.3]$. Customers select between 1 and 10 cities and specify different image resolutions to meet their needs, with payments randomised between $1$ and $1.2$ times the standard price. The discount factor is modelled by the exponential function $\delta(d^k_s)=0.5e^{\left(-\left(\frac{|Y|}{|X|}\right)\cdot(d^k_s - 1)\right)}+0.4$, where $|Y|$ and $|X|$ represent the total service demands and the number of services offered, respectively, and $d^k_s$ signifies the demand for service $s$ in task $k$. The maximum discount is capped at $60\%$. We employ a clustering algorithm based on Jacobian distance (see Algorithm $1$ in Section $5.1$ below) to group customers by their service needs, addressing the extensive customer base. For team formation, companies within a social network are organised into teams using a greedy algorithm (see Algorithm $2$ in Section $5.2$ below) that forms teams based on a similarity threshold. This algorithm clusters companies by extending connections among those below this similarity threshold. Tasks generated through this method are allocated to the most competitively priced teams (see Algorithm $3$ in Section $5.3$ below). Subsequent subsections will adjust the number of customers and similarity thresholds to explore various experimental scenarios and analyse the results.

\begin{algorithm}
\caption{Customer and Companies Definition}\label{user_provider_generation}
\begin{algorithmic}[1]
\State \textbf{Input:} Customers $X$ and Companies $Y$
\State \textbf{Output:} Lists of Customers $X$ and Companies $Y$ with associated information
\State Read the city with geographical data from a CSV file in Basilisk
\State Initialise Positions $P$ list with city name ($p.\text{name}$), population ($p.\text{population}$), and area ($p.\text{area}$)
\For{$x \in |X|$}
    \State Generate a random number of needed services as list $positions$
    \State Random $x$ price factor $pf_x \in [1,1.2]$
    \State $x.\text{positions} \gets positions$
    \For{$i \in positions$}
        \State $x.\text{service\_resolution} \in \{low, medium, high\}$
        \State $x.\text{service\_cost}[i] \gets 10 \times position[i].\text{area} \times pf_x$
    \EndFor
    \State $x.\text{total\_payment} \gets \sum(x.\text{service\_cost})$
    \State Append $x$ to customers' list $X$
\EndFor
\For{$y \in |Y|$}
    \State Random generates a Low Earth Orbit (LEO) satellite with orbital parameters and calculates the available services $services$
    \State Random $y$ price factor $pf_y \in [1,1.3]$
    \State $y.\text{services} = services$
    \For{$i \in services$}
        \State $y.\text{service\_resolution} \subseteq \{low, medium, high\}$
        \State $y.\text{service\_cost}[i] \gets 3 \times position[i].\text{area} \times pf_y$
    \EndFor
    \State Operation cost $OC_y \gets |services| \times 5000$
    \State $y.\text{total\_cost} \gets OC_y + \sum(y.\text{service\_cost})$
    \State Append $y$ to companies' list $Y$
\EndFor
\State \Return Lists of Customers $X$ and Companies $Y$
\end{algorithmic}
\end{algorithm}

\subsection{Customer and Companies Definition}

Algorithm \ref{user_provider_generation} is designed to define lists of customers and companies, each with specific attributes. Initially, it reads geographical data from a CSV file to establish a list of potential locations, including city names, populations, and areas. The algorithm randomly generates the required services for each customer, assigns geographical positions, decides resolutions, calculates service prices, and aggregates total payments. These details are then compiled into a customer profile appended to the customer list. The algorithm uses the Basilisk simulation platform for companies to simulate Low Earth Orbit (LEO) satellites with defined orbital parameters, determines the available services, and computes operation and service costs. Each company profile, including these cost details, is added to a company list. The output of this algorithm comprises two lists: one for customers (denoted as $X$) and one for companies (denoted as $Y$), containing all the relevant information as defined in the main paper.

\section{Algorithms in Our Experiments}

In this section, we discuss the main algorithms featured in the article, which encompass a variety of computational strategies. Initially, we used the Basilisk simulation platform \cite{kenneally2020basilisk,wood2018flexible} to generate data for consumers and companies, providing comprehensive information for each entity. Subsequently, we employed a clustering algorithm to create tasks by grouping related activities based on shared characteristics. We also implemented a greedy algorithm for team formation, which optimised group assembly based on criteria such as skill levels and task requirements. Lastly, task allocation was managed using a minimum price strategy, assigning tasks to the lowest bidder. These methodologies facilitate efficient operations and strategic planning within the simulated environment. All experiments run on a MacBook Pro with an M1 Max chip and 32 GB Memory.      

\begin{algorithm}
\caption{Clustering Algorithm for Task Generation}\label{clustering}
\begin{algorithmic}[1]
\State \textbf{Input:} Customers $X$ and number of tasks $N$
\State \textbf{Output:} Set of tasks $K$, where $|K| = N$
\State Initialize $N$ centroids randomly from $X$
\While{centroids change}
    \For{each customer $x \in X$}
        \State Calculate Jacobian distance $d_J(x, k)$ for each centroid $k$
        \State $d_J(x, k)=1-\frac{|x.services\bigcap k.services|}{|x.services\bigcup k.services|}$
        \State Assign $x$ to the task with the nearest centroid
    \EndFor
    \For{each task $k \in K$}
        \State Update centroid of $k$ based on members
    \EndFor
\EndWhile
\State \Return Tasks $K$
\end{algorithmic}
\end{algorithm}

\subsection{Task Generation Algorithm}

Algorithm \ref{clustering} is a clustering-based approach specifically tailored for generating tasks by grouping customers. It takes two primary inputs: a set of customers $X$ and a desired number of tasks $N$, which are $200$, $400$, $600$, $800$ and $1000$ in our paper. Initially, $N$ centroids are randomly selected from $X$ to represent the starting point of each task. The algorithm proceeds iteratively; in each iteration, each customer $x\in X$ is assigned to a task based on the Jacobian distance, a measure of similarity especially relevant for datasets with binary attributes—between $x$ and each centroid. Customers are assigned to the nearest centroid, thus forming preliminary groups. After all customers have been assigned, each task's centroid is recalculated based on the current members of that task, which typically involves computing a statistical mean or median of the assigned customers' attributes. This process repeats until the centroids stabilise and no longer change, indicating that the tasks have been sufficiently optimised. The output is a set of $N$ distinct tasks, each characterised by a group of customers who are similar to each other according to the Jacobian distance metric. This algorithm organises customers into coherent groups for targeted task assignments in various applications such as marketing, service delivery optimisation, or resource allocation. After Algorithm \ref{clustering} generates a new task, each user calculates a new payment amount based on the discount function.

\begin{algorithm}
\caption{Team Formation based on Similarity}\label{greedy}
\begin{algorithmic}[1]
\State \textbf{Input:} A list of companies $Y$ and the similarity for team formation
\State \textbf{Output:} List of teams $M$
\State Initialise an empty list $M$
\State Initialise a boolean array $Allocated\_company$ with all values set to \textit{False}
\State $Two\_company\_Similarity(i, j)$
\State \Return $similarity = \frac{|S_i\bigcap S_j|}{|S_i|}$
\State Initialise $random\_list[|Y|]$ randomly from $Y$
\For{$i = 0$ \textbf{to} $length(random\_list) - 1$}
    \If{\textbf{not} $Allocated\_company[i]$}
        \State $new\_team \gets [random\_list[i]]$
        \State $Allocated\_company[i] \gets$ \textit{True}
        \For{$j = 0$ \textbf{to} $length(random\_list) - 1$}
            \If{\textbf{not} $Allocated\_company[j]$\\\textbf{and} $Two\_company\_Similarity(i, j) \leq similarity$}
                \State Append $random\_list[j]$ to $new\_team$
                \State $Allocated\_company[j] \gets$ \textit{True}
            \EndIf
        \EndFor
        \State Append $new\_team$ to $M$
    \EndIf
\EndFor
\State \Return $M$
\end{algorithmic}
\end{algorithm}

\subsection{Team Formation Algorithm}

In our experimental tests, we divided the similarity of the formed teams into 300 intervals from 0 to 0.6, and the algorithm looped from 0 to 0.6. Algorithm \ref{greedy} proposed a method for team formation of companies based on their similarities of services, that companies within the same team meet a predefined similarity threshold. Initially, a list of companies and a similarity as inputs. It starts by creating an empty list to store teams and a boolean array to track whether each company has already been allocated to a team. The list of companies is shuffled to randomise the team formation process. For each company in the shuffled list that has not been allocated, a new team is started with that company. The algorithm then checks all other unallocated companies to see if their similarity with the initial company meets the threshold using a predefined similarity function. A company qualifies and is added to the current team and marked as allocated. This process is repeated until all companies have been allocated to a team. The result is a list of teams, each consisting of companies similar to each other according to the specified threshold.

\begin{algorithm}\footnotesize
\caption{Task Allocation and Minimum Offer Selection}\label{alg:TaskAllocationAndMinOffer}
\begin{algorithmic}[1]
\State \textbf{Input:} List of Tasks $K$ and Teams $M$
\State \textbf{Output:} Allocation results $TA\_tasks$ and number of failed customers $Failed\_customers$
\Function{Find\_Min\_Offer}{$k$, $M$}
    \State $min\_offer=\infty$
    \State $allocated\_team=-1$
    \For{$i$ in range of $M$}
        \State $offer=0$
        \If{$k.positions \subseteq M[i].services$}
            \For{$j$ in range of $k.positions$}
                \For{$q$ in range of $M[i].services$}
                    \If{$k.positions[j]$=$M[i].services[q]$}
                        \State $offer+= M[i].service\_cost[q]$
                    \EndIf
                \EndFor
            \EndFor
            \If{$offer\leq min\_offer$}
                \State $min\_offer=offer$
                \State $allocated\_team=M[i].id$
            \EndIf
        \EndIf
    \EndFor
    \State \Return $min\_offer$, $allocated\_team$
\EndFunction
\Procedure{Task\_Allocation}{$K$, $M$}
    \State $TA\_tasks=[]$
    \State $Failed\_customers=0$
    \For{$k$ in $K$}
        \State $min\_offer, team \gets Find\_Min\_Offer(k, M)$
        \State $TA\_tasks[k]=team$
        \If{$TA\_tasks[k]=-1$}
            \State $Failed\_customers+=k.num\_customers$
        \EndIf
    \EndFor
    \State return $TA\_tasks$ and $Failed\_customers$
\EndProcedure
\end{algorithmic}
\end{algorithm}

\subsection{Task Allocation}

Algorithm \ref{alg:TaskAllocationAndMinOffer} proposed a method that matches tasks with teams based on service needs and costs. The $Find\_Min\_Offer$ function identifies the cheapest team for each user by comparing available offers and selecting the one with the lowest total cost. The $Task\_Allocation$ function oversees the assignment process, calculating total task payments and revenue generated while tracking the number of filed customers. In this section, we show a detailed explanation of the experimental results under five scenarios mentioned in the main paper, which were conducted under customers and companies setting in the main paper and include five task volumes: $200$, $400$, $600$, $800$, and $1000$. 

\section{Additional Experimental Results}

\subsection{Additional Results for Scenario 1}

In our main paper, we conducted ten trials for each task count to verify the stochastic generation of users and companies. \textbf{Fig. \ref{STD}} displays the standard deviation for ten datasets at various levels of similarity within Scenario 1, highlighting the variability from trial to trial.The x-axis measures the similarity between teams—derived from Algorithm 2—ranging from $0.0$ to $0.45$, while the y-axis represents the average revenue, scaled to tens of millions of dollars. Different lines indicate the outcomes for $400$ tasks (blue) and $1000$ tasks (red), with surrounding shaded regions illustrating the standard deviation across the ten identical parameter iterations. 

\begin{figure}[ht]
    \centering
    \includegraphics[scale=0.23]{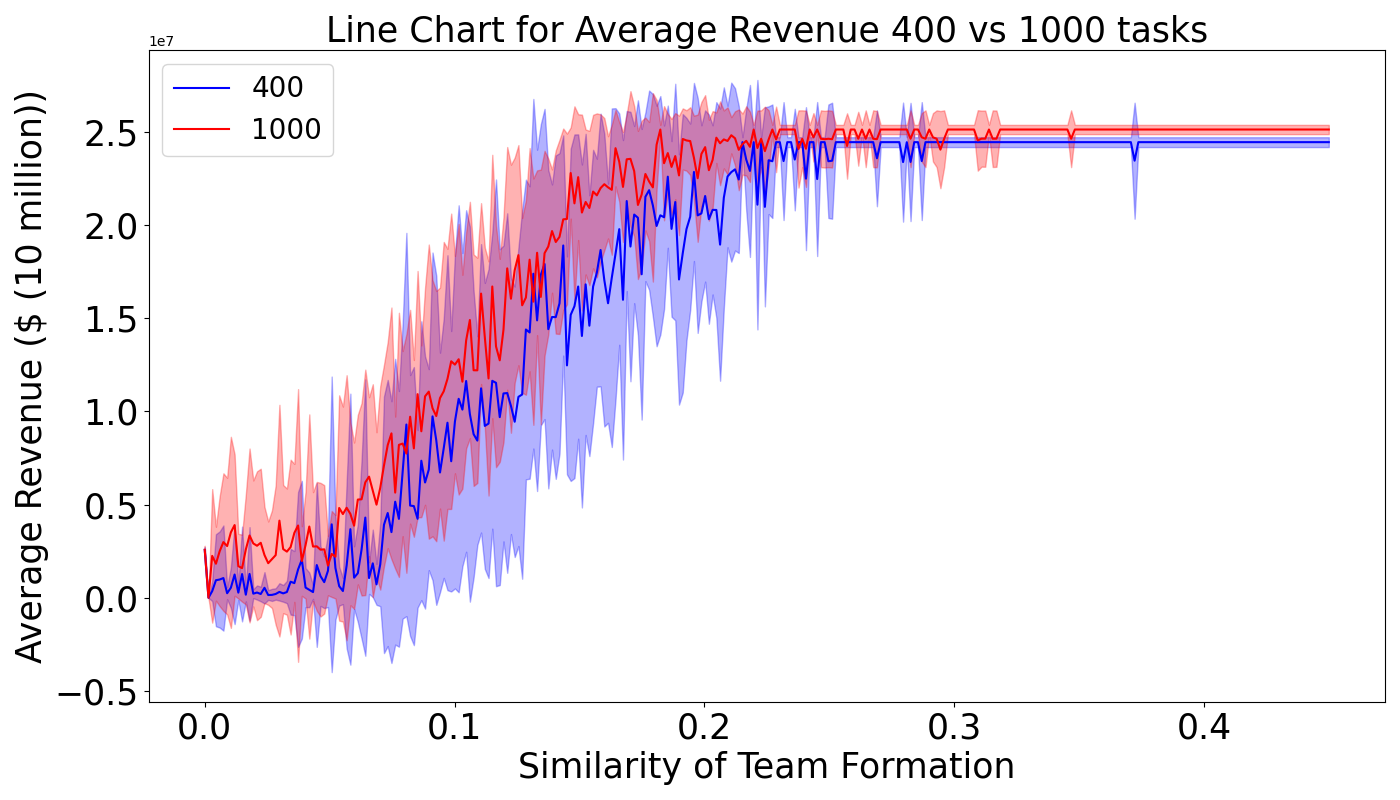}
    \caption{Standard derivation with $10$ random tests in $400$ and $1000$ tasks (Scenario 1)}\label{STD}
\end{figure}

Notably, average revenue increases with greater similarity for both task counts, yet the returns from $1000$ tasks consistently exceed those from $400$ tasks across the entire range of similarity. Both revenue lines plateau at higher similarity levels, suggesting a saturation point where further increases in similarity no longer enhance average revenues. Additionally, the initially broad standard deviations narrow as similarity increases, as shown by the contracting shaded areas and decreasing variance, indicating a stabilisation of average revenues across both task quantities.

\begin{figure}[ht]
    \centering
    \includegraphics[scale=0.23]{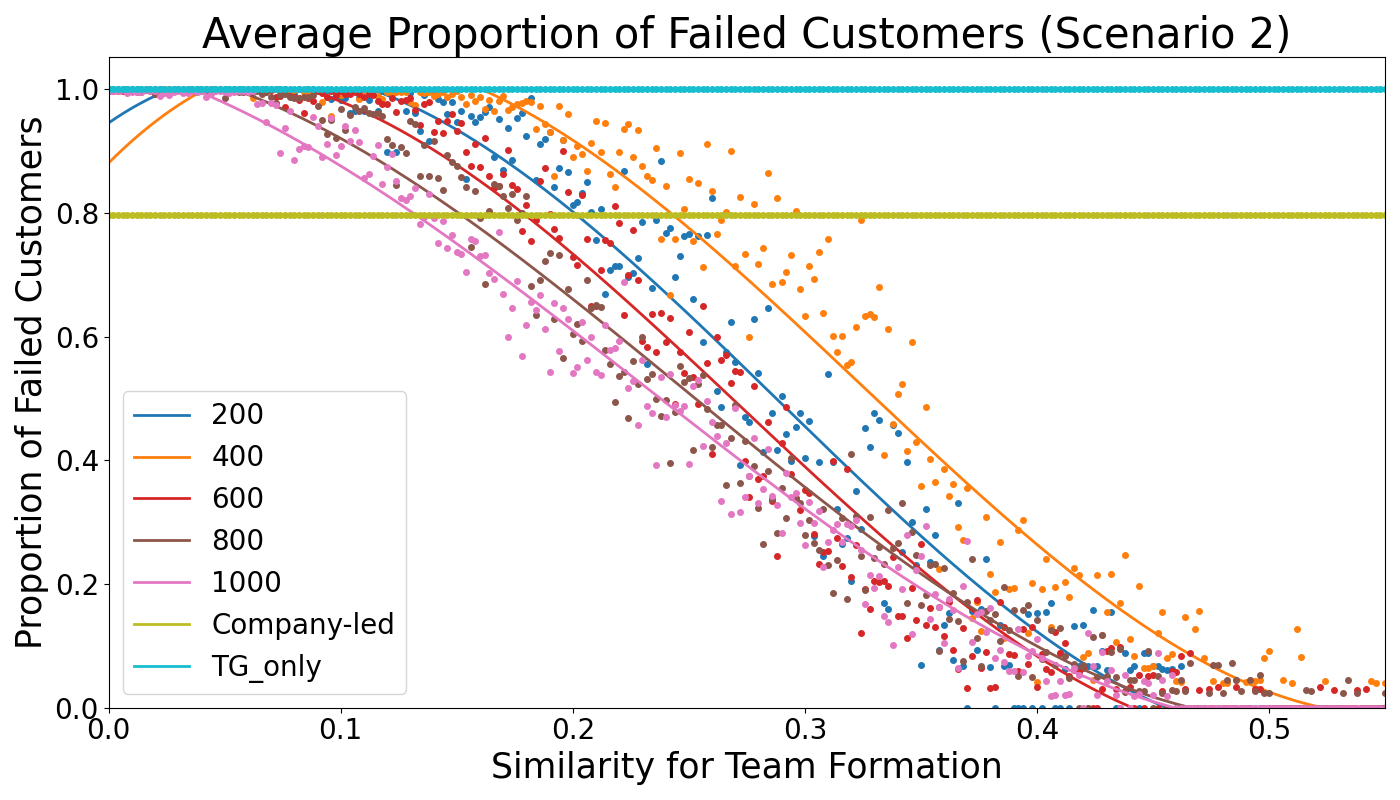}
    \caption{Average proportion of failed customers (Scenario 2)}\label{failed_customers_s2}
\end{figure}

\begin{table}[ht]
    \footnotesize
    \centering
    \setlength{\tabcolsep}{0.5pt}
    \caption{Average customers' payment and companies' revenue at Nash equilibrium in 200-1000 tasks for scenario 1}
    \label{NE_table}
    \begin{tabular}{llllll}
        \hline
        \textbf{Tasks}&\textbf{Similarity} & \textbf{Utilities} (\$) & \textbf{Company-led} & \textbf{Customer-led} & \textbf{PRF}\\
        \hline
        200 & 0.364 & Payment & 172,713 & 141,500 &  -18\%    \\[2pt]
            &       & Revenue  & 2,612,353 & 23,583,488& +900\%  \\[2pt]
        400 & 0.334 & Payment & 173,031 & 146,742  & -15\%    \\[2pt]
            &       & Revenue & 2,565,647 & 24,457,112 & +953\%  \\[2pt]
        600 & 0.322 & Payment & 170,786 & 147,020  & -14\%    \\[2pt]
            &       & Revenue & 2,566,651 & 24,503,497 &  +954\% \\[2pt]
        800 & 0.312 & Payment & 172,753 & 149,964  & -13\%    \\[2pt]
            &       & Revenue & 2,521,374 & 24,994,130 & +991\%    \\[2pt]
        1000& 0.304 & Payment & 172,713 & 150,816    &  -12.5\%  \\[2pt]
            &       & Revenue & 2,578,431 & 25,136,018  &  +974\% \\
        \hline
    \end{tabular}
\end{table}

\textbf{Table~\ref{NE_table}} displays the NE states of companies across various clustering scenarios. As the number of tasks increases, a lower similarity threshold is needed to reach equilibrium due to tasks spanning fewer locations, simplifying fulfilment. From a payment perspective, fewer tasks lead to reduced average customer payments by up to $18\%$ compared to "Company-led", as fewer tasks boost demand for the same services within those tasks, enabling higher discounts. Consequently, this payment reduction decreases overall market payments and the company revenue. In the "Customer-led", company revenue surges by up to $991\%$ because previously unmanageable, high-value tasks become feasible, thus completing more lucrative tasks and substantially increasing revenue.

\subsection{Results for Scenario 2}

In Scenario 2, we reduced the number of companies from 30 to 15, while other parameters remained unchanged compared to Scenario 1. \textbf{Fig. \ref{failed_customers_s2}} shows the percentage of customers with uncompleted tasks. The yellow line (labelled ``Customer-led") shows that approximately $80\%$ of customers failed to complete their tasks. Coloured data points represent the distribution of customer task completions, categorised by different task numbers and increasing company similarity. It is observed that the incomplete rates for all tasks fall below the yellow line when similarity exceeds 0.25 and approach a $100\%$ success rate as similarity increases. However, at the same similarity level, a lower number of tasks correlates with a higher percentage of uncompleted tasks. For example, at a similarity of 0.25, about $60\%$ of customers in 200 tasks remain uncompleted, compared to $50\%$ in 1000 tasks.

\begin{figure}[ht]
    \centering
    \includegraphics[scale=0.23]{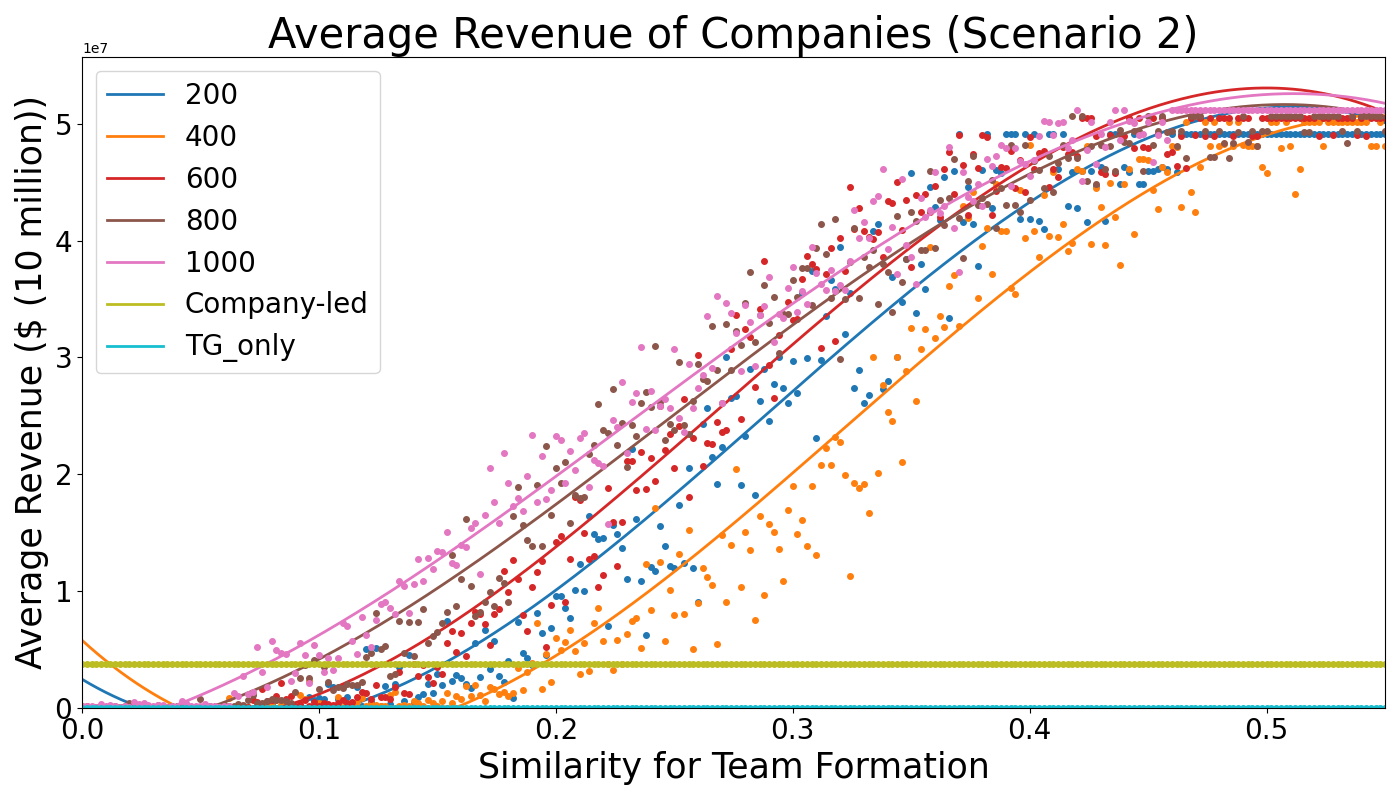}
    \caption{Average revenue of companies (Scenario 2)}\label{profit_s2}
\end{figure}

\textbf{Fig.~\ref{profit_s2}} displays polynomial fitting functions for companies' average revenue as a function of task volumes, with each scenario marked by differently coloured data points. The yellow line, labelled ``Customer-led", serves as a standard reference, while the light blue line represents the scenario where customers generate tasks without company team formation. The figure demonstrates that higher similarity for team formation significantly increases average revenues for companies. Revenues exceed the reference line once similarity surpasses 0.2. Moreover, at a fixed similarity, more tasks yield higher average revenues; for instance, at 0.25, revenues from 1000 tasks are about 1.5 times those from 400 tasks. Notably, the rate at which completion rates converge towards higher levels is noticeably slower for 400 tasks compared to 200 tasks in Scenario 2.

\textbf{Table~\ref{NE_s2}} presents the experimental results for the NE state across different tasks in Scenario 2, assessing financial performance through payments and revenues. The similarity index for each task volume ranges from 0.46 to 0.502, indicating the degree of similarity within each set of tasks. For example, with 200 tasks, payments decreased by $15.4\%$ in the test scenario, while revenues surged by $1254\%$. This trend of reduced payments and markedly increased revenues is consistent across all task volumes, suggesting that the test scenario (``Customer-led") is financially more beneficial than the baseline (``Company-led"). These substantial increases in revenue, often exceeding $1000\%$, highlight a dramatic improvement in efficiency or cost-effectiveness. The reduction in the number of companies significantly increases the difficulty of team formation. The similarity required to reach equilibrium is almost consistent for the $400$, $600$, and $800$ tasks, indicating that a change in the number of tasks does not affect the equilibrium due to the scarcity of companies. However, for both $200$ and $1000$ tasks, the similarity needed to achieve equilibrium is the same, suggesting that having too many or too few tasks makes it relatively easier to reach equilibrium.

\begin{table}[ht]
    \centering\footnotesize
    \setlength{\tabcolsep}{0.5pt}
    \begin{tabular}{llllll}
        \hline
        \textbf{Tasks} & \textbf{Similarity} & \textbf{Utilities} (\$) & \textbf{Company-led} & \textbf{Customer-led} & \textbf{PRF}\\
        \hline
        200 & 0.462 & Payment & 172,681  & 147,413    &  -15.4\% \\[2pt]
            &       & Revenue  & 3,918,280 & 49,137,997 & +1254\%  \\[2pt]
        400 & 0.502 & Payment & 172,684 & 150,412  & -12.9\%     \\[2pt]
            &       & Revenue  & 3,883,929 & 50,137,510 & +1290\%  \\[2pt]
        600 & 0.498 & Payment & 172,491 & 151,503  & -12.2\%    \\[2pt]
            &       & Revenue & 3,542,108 & 50,501,053 &  +1425\% \\[2pt]
        800 & 0.502 & Payment & 172,185 & 151,941  & -11.7\%    \\[2pt]
            &       & Revenue & 3,788,766 & 50,647,273 & +1336\%    \\[2pt]
        1000& 0.46 & Payment & 173,427 & 153,597    &  -11.4\%  \\[2pt]
            &       & Revenue & 3,580,109 & 51,199,116  &  +1430\% \\[2pt]
        \hline
    \end{tabular}
    \caption{Average customers' payment and companies' revenue at Nash equilibrium in 200-1000 tasks for scenario 2}
    \label{NE_s2}
\end{table}

\subsection{Results for Scenario 3}

In Scenario 3, we reduced the number of companies from 30 to 15 and increased the range of services offered by companies from $[10, 30]$ to $[30, 50]$, while other parameters remain unchanged compared to Scenario 1.

\begin{figure}[ht]
    \centering
    \includegraphics[scale=0.23]{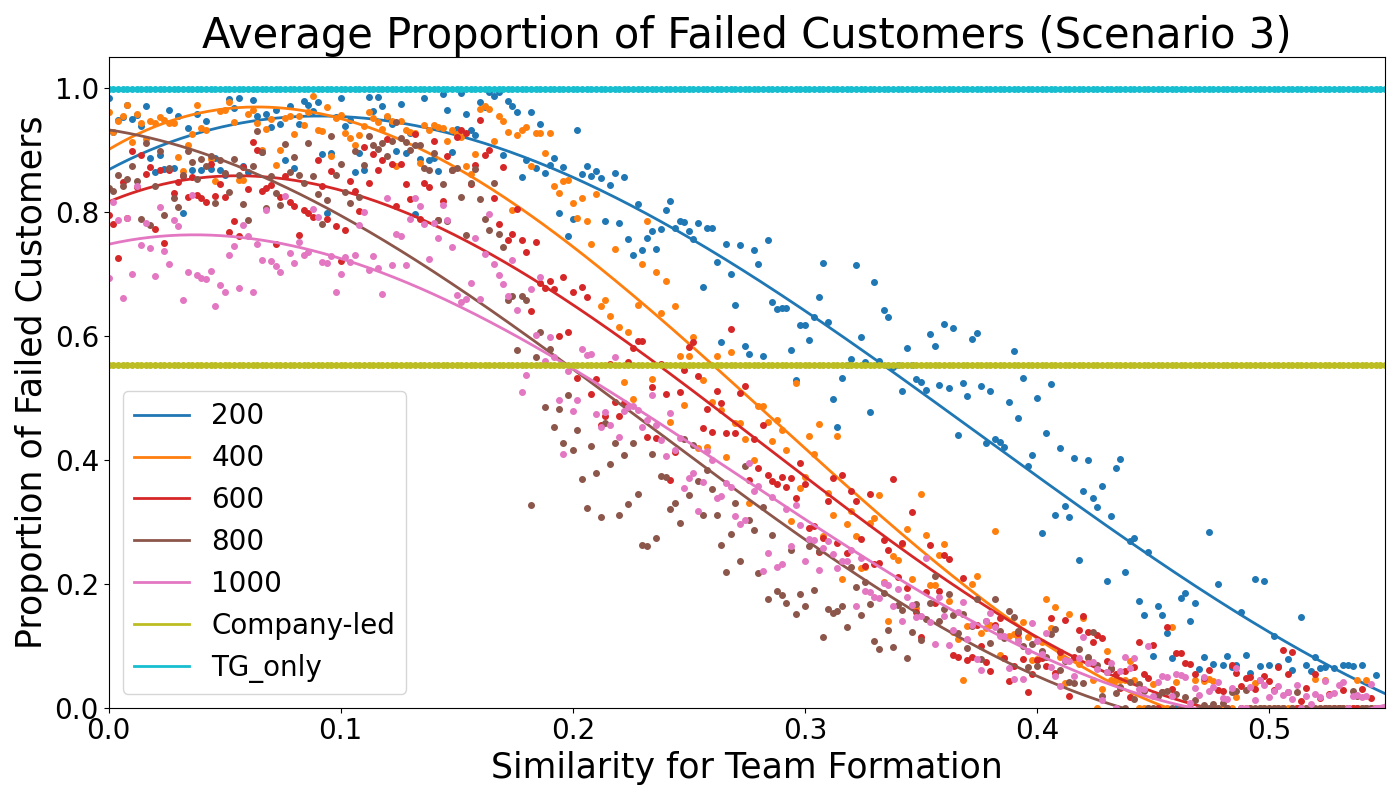}
    \caption{Average proportion of failed customers (Scenario 3)}\label{failed_customers_s3}
\end{figure}

\textbf{Fig.~\ref{failed_customers_s3}} shows the percentage of customers with uncompleted tasks. The yellow line (labelled ``Company-led") shows that approximately $55\%$ of customers failed to complete their tasks, which is significantly reduced compared with Scenario 1 and 2. It is observed that the incomplete rates for all tasks fall below the yellow line when similarity exceeds 0.33 and approach a $100\%$ success rate as similarity increases. However, at the same similarity level, a lower number of tasks correlates with a higher percentage of uncompleted tasks. For example, at a similarity of 0.3, about $70\%$ of customers in 200 tasks remain uncompleted, compared to $30\%$ in 1000 tasks.

\begin{figure}[ht]
    \centering 
    \includegraphics[scale=0.23]{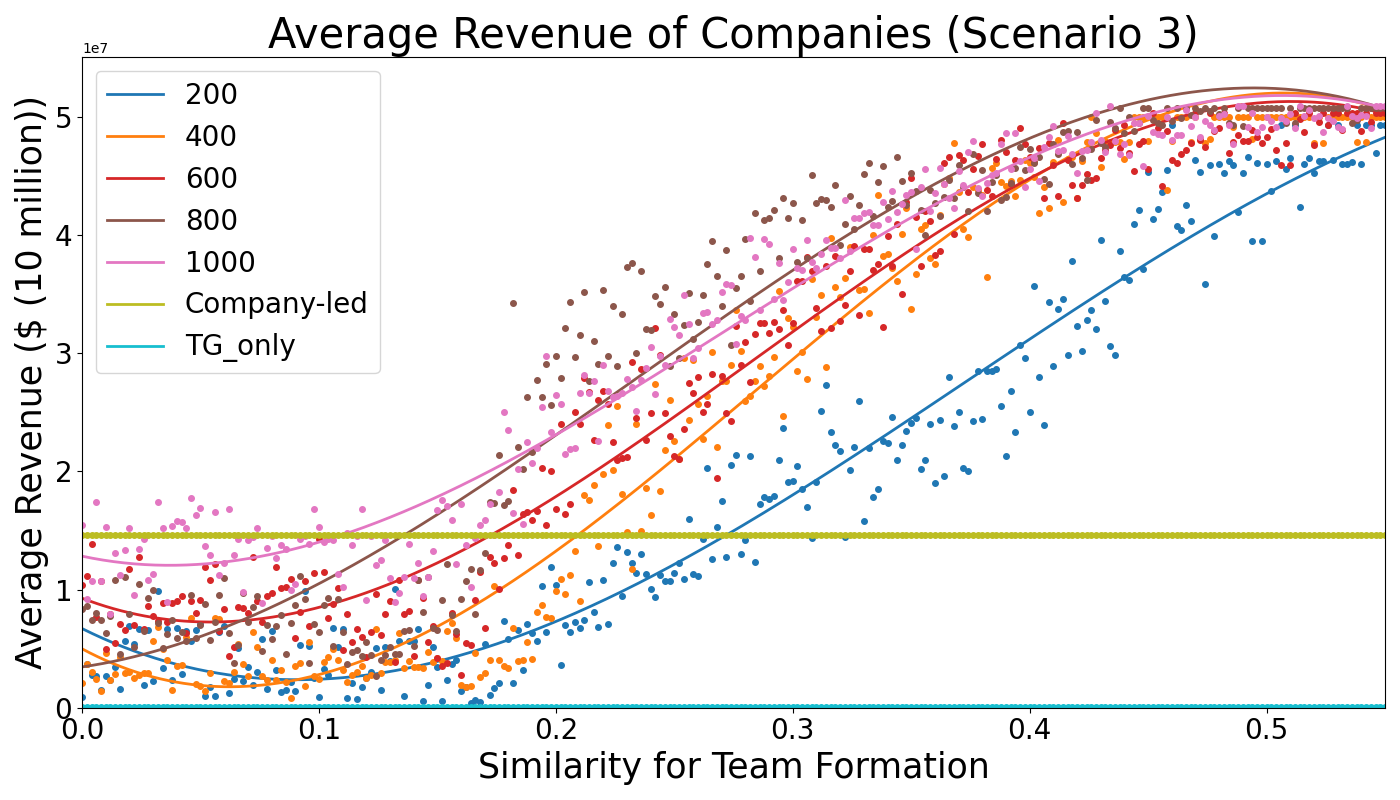}
    \caption{Average revenue of companies (Scenario 3)}\label{profit_s3}
\end{figure}

\textbf{Fig.~\ref{profit_s3}} displays polynomial fitting functions for companies' average revenue as a function of task volumes, with each scenario marked by differently coloured data points. The yellow line, labelled ``Company-led", serves as a standard reference, while the light blue line represents the scenario where customers generate tasks without company team formation. The figure demonstrates that higher similarity for team formation significantly increases average revenues for companies. Revenues exceed the reference line once similarity surpasses 0.25. Moreover, at a fixed similarity, more tasks yield higher average revenues; for instance, at 0.3, revenues from 1000 (around $3.8*10^7$) tasks are about $2$ times those from 400 tasks (around $1.8*10^7$). It is worth noting that in Scenario $3$, the trend observed in the task experiments is similar to Scenario $1$: as the number of tasks decreases, the difficulty of forming teams increases. The similarity required to reach equilibrium for $200$ tasks is significantly higher than in other tests. Slightly different from Scenario $1$, when exceeding the reference, 800 tasks reach equilibrium earlier relative to $1000$ tasks.

\begin{table}[ht]
    \centering\footnotesize
    \setlength{\tabcolsep}{0.5pt}
    \begin{tabular}{llllll}
        \hline
        \textbf{Tasks} & \textbf{Similarity} & \textbf{Utilities} (\$) & \textbf{Company-led} & \textbf{Customer-led} & \textbf{PRF}\\
        \hline
        200 & 0.538 & Payment & 173,392 & 148,033 &  -15\%    \\[2pt]
            &       & Revenue  & 14,435,611 & 49,344,631& +341\%  \\[2pt]
        400 & 0.496 & Payment & 172,128 & 149,990  & -13\%    \\[2pt]
            &       & Revenue & 15,044,173 & 49,996,915 & +332\%  \\[2pt]
        600 & 0.47 & Payment & 171,852 & 151,057  & -12.2\%    \\[2pt]
            &       & Revenue & 14,478,544 & 50,352,640 &  +347\% \\[2pt]
        800 & 0.448 & Payment & 172,431 & 152,215  & -11.8\%    \\[2pt]
            &       & Revenue & 14,137,554 & 50,738,408 & +358\%    \\[2pt]
        1000& 0.48 & Payment & 172,566 & 152,874    &  -11.5\%  \\[2pt]
            &       & Revenue & 15,011,309 & 50,958,118  &  +339\% \\[2pt]
        \hline
    \end{tabular}
    \caption{Average customers' payment and companies' revenue at Nash equilibrium in 200-1000 tasks for scenario 3}
    \label{NE_s3}
\end{table}

\textbf{Table~\ref{NE_s3}} lists the experimental results of the Nash Equilibrium (NE) state for different tasks in Scenario 2, evaluating financial performance through payments and revenues. The similarity index for each task volume ranges from 0.48 to 0.538, indicating the level of similarity within each set of tasks. For example, in the test scenario, payments for 200 tasks decreased by $15\%$, while revenues surged by $341\%$. This trend of decreasing payments and significantly increasing revenues is consistent across all task volumes, suggesting that the test scenario (``Customer-led") is more financially beneficial than the baseline scenario (``Company-led"). Unlike Scenarios 1 and 2, when the number of tasks changes, the revenues do not increase significantly but fluctuate between $330\%$ and $350\%$. We can claim that the entire market is in a Stackelberg Equilibrium state at $200$ tasks with a team formation similarity of 0.36.

\subsection{Results for Scenario 4}

In Scenario 4, we reduced each consumer's needs from $[1, 10]$ to $[1, 5]$, while other parameters remained unchanged compared to Scenario 1. This change shifts each consumer's average task from complex (average $5$ cities) to simple (average $2.5$ cities) and the average payment amount from high (around $\$170,000$) to low (around $\$94,000$). 

\begin{figure}[ht]
    \centering
    \includegraphics[scale=0.23]{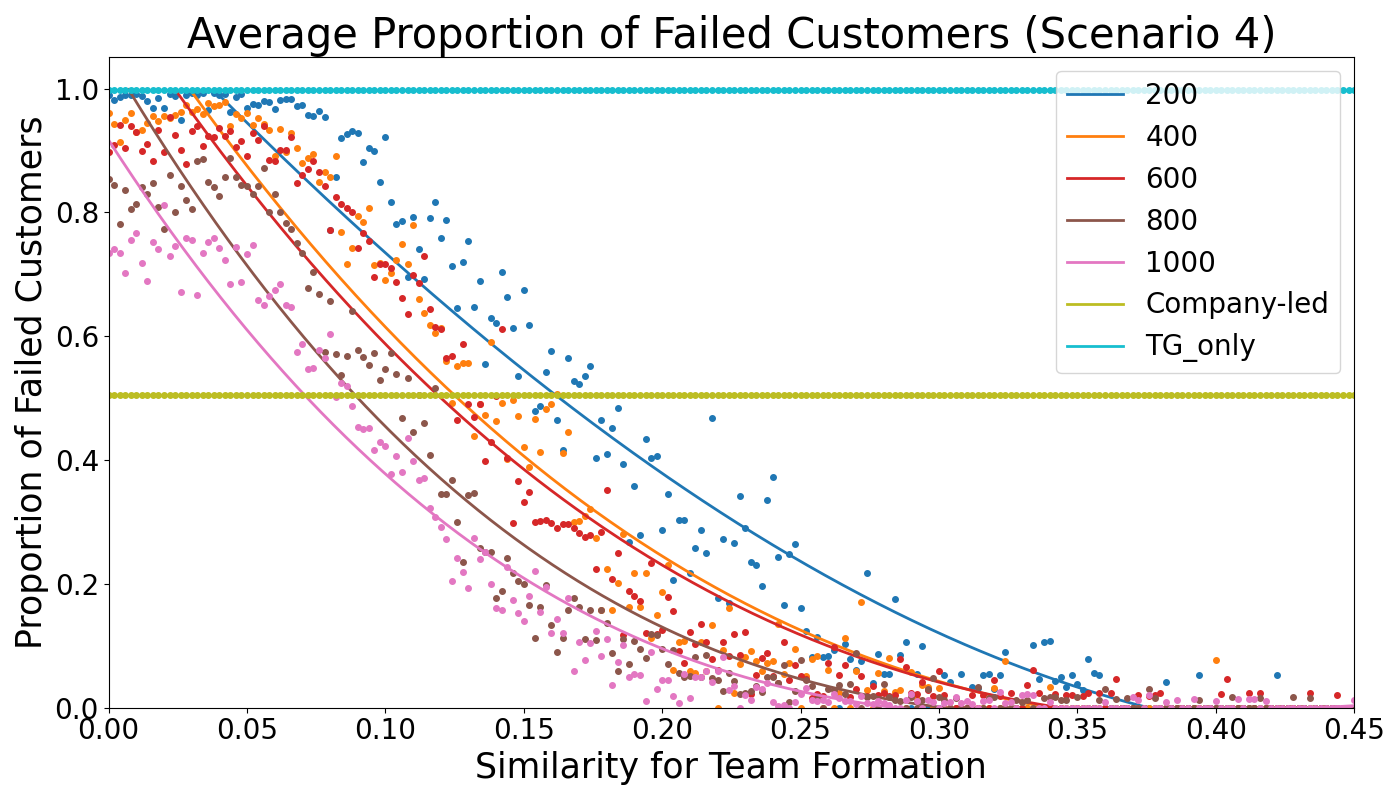}
    \caption{Average proportion of failed customers (Scenario 4)}\label{failed_customers_s4}
\end{figure}

\textbf{Fig.~\ref{failed_customers_s4}} shows the percentage of customers with uncompleted tasks. The yellow line (labelled ``Customer-led") shows that approximately $50\%$ of customers failed to complete their tasks, which is significantly reduced compared with Scenario 1 and 2. It is observed that the incomplete rates for all tasks fall below the yellow line when similarity exceeds $0.18$ and approach a $100\%$ success rate as similarity increases. However, at the same similarity level, a lower number of tasks correlates with a higher percentage of uncompleted tasks. For example, at a similarity of 0.2, about $40\%$ of customers in 200 tasks remain uncompleted, compared to $10\%$ in 1000 tasks.

\begin{figure}[ht]
    \centering
    \includegraphics[scale=0.23]{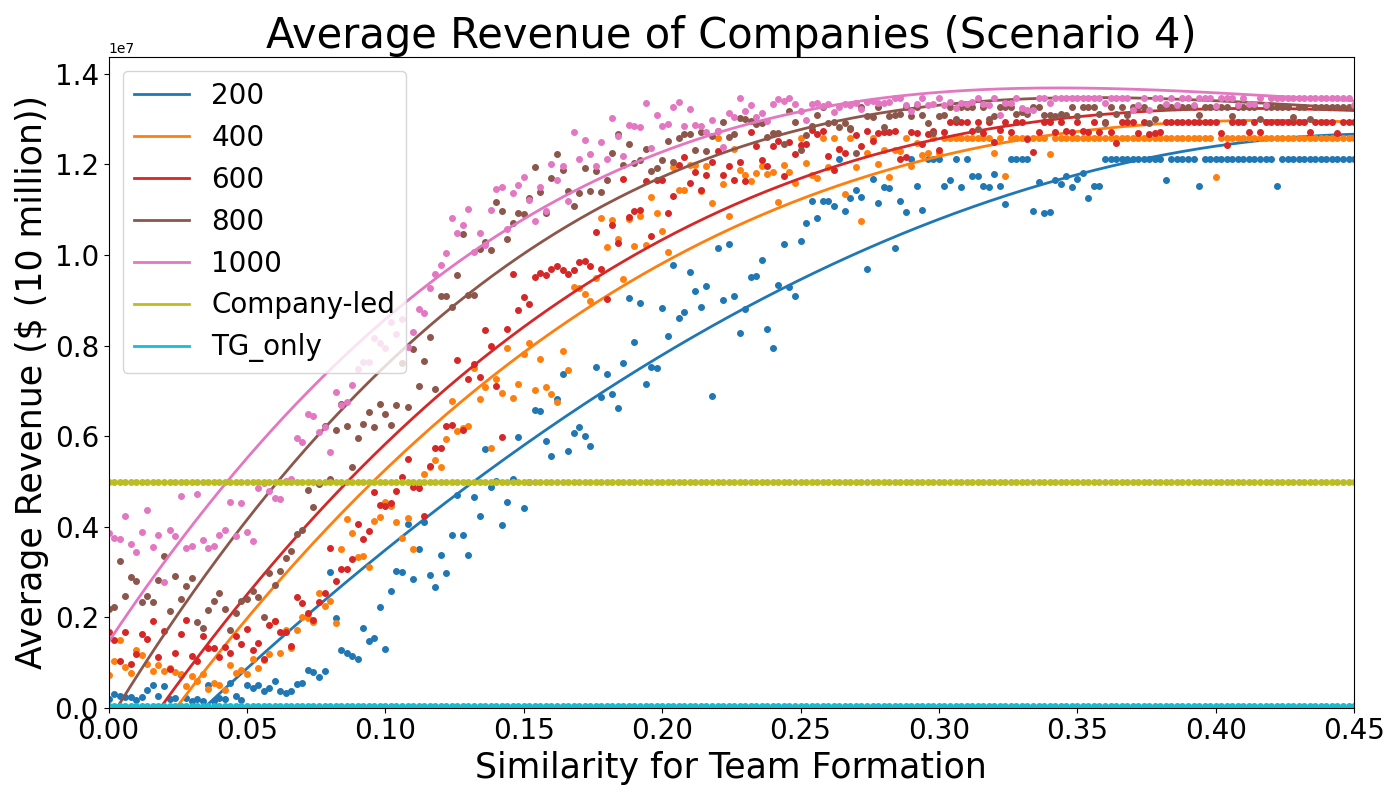}
    \caption{Average revenue of companies (Scenario 4)}\label{profit_s4}
\end{figure}

\textbf{Fig.~\ref{profit_s4}} displays polynomial fitting functions for the average revenue of companies as a function of task volumes, with each scenario marked by differently coloured data points. The yellow line, labelled ``Company-led," serves as a standard reference, while the light blue line represents the scenario where customers generate tasks without company team formation. The figure demonstrates that higher similarity for team formation significantly increases average revenues for companies. Revenues exceed the reference line once similarity surpasses 0.15. Moreover, at a fixed similarity, more tasks yield higher average revenues; for instance, at 0.2, revenues from 1000 (around $1.2*10^7$) tasks are about $2$ times those from 400 tasks (around $0.6*10^7$). 

\begin{table}[ht]
    \centering\footnotesize
    \setlength{\tabcolsep}{0.5pt}
    \begin{tabular}{llllll}
        \hline
        \textbf{Tasks} & \textbf{Similarity} & \textbf{Utilities} (\$) & \textbf{Company-led} & \textbf{Customer-led} & \textbf{PRF}\\
        \hline
        200 & 0.36 & Payment & 94,196 & 72,760 &  -21\%    \\[2pt]
            &       & Revenue  & 5,060,874 & 12,126,813& +239\%  \\[2pt]
        400 & 0.326 & Payment & 93,922 & 75,550  & -20.5\%    \\[2pt]
            &       & Revenue & 5,272,500 & 12,591,740 & +239\%  \\[2pt]
        600 & 0.304 & Payment & 93,843 & 77,673  & -17.3\%    \\[2pt]
            &       & Revenue & 4,917,013 & 12,945,587 &  +263\% \\[2pt]
        800 & 0.288 & Payment & 94,229 & 79,608  & -15.6\%    \\[2pt]
            &       & Revenue & 4,668,404 & 13,268,062 & +284\%    \\[2pt]
        1000& 0.274 & Payment & 94,422 & 80,815 &  -14.5\%  \\[2pt]
            &       & Revenue & 5,038,154 & 13,469,243  &  +267\% \\[2pt]
        \hline
    \end{tabular}
    \caption{Average customers' payment and companies' revenue at Nash equilibrium in 200-1000 tasks for scenario 4}
    \label{NE_s4}
\end{table}
\textbf{Table~\ref{NE_s4}} lists the experimental results of the NE state for different tasks in Scenario 4, evaluating financial performance through payments and revenues. The similarity index for each task volume ranges from 0.274 to 0.36, indicating the level of similarity within each set of tasks. For example, in the test scenario, payments for 200 tasks decreased by $21\%$, while revenues surged by $239\%$. This trend of decreasing payments and significantly increasing revenues is consistent across all task volumes, suggesting that the test scenario (``Customer-led") is more financially beneficial than the baseline scenario (``Company-led"). 

\subsection{Results for Scenario 5}

In Scenario $5$, we increased the number of consumers from 5,000 to 10,000, while other parameters remain unchanged compared to Scenario 1. This change has increased the total payment amount in the market, potentially increasing the revenues for companies. 

\begin{figure}[ht]
    \centering
    \includegraphics[scale=0.23]{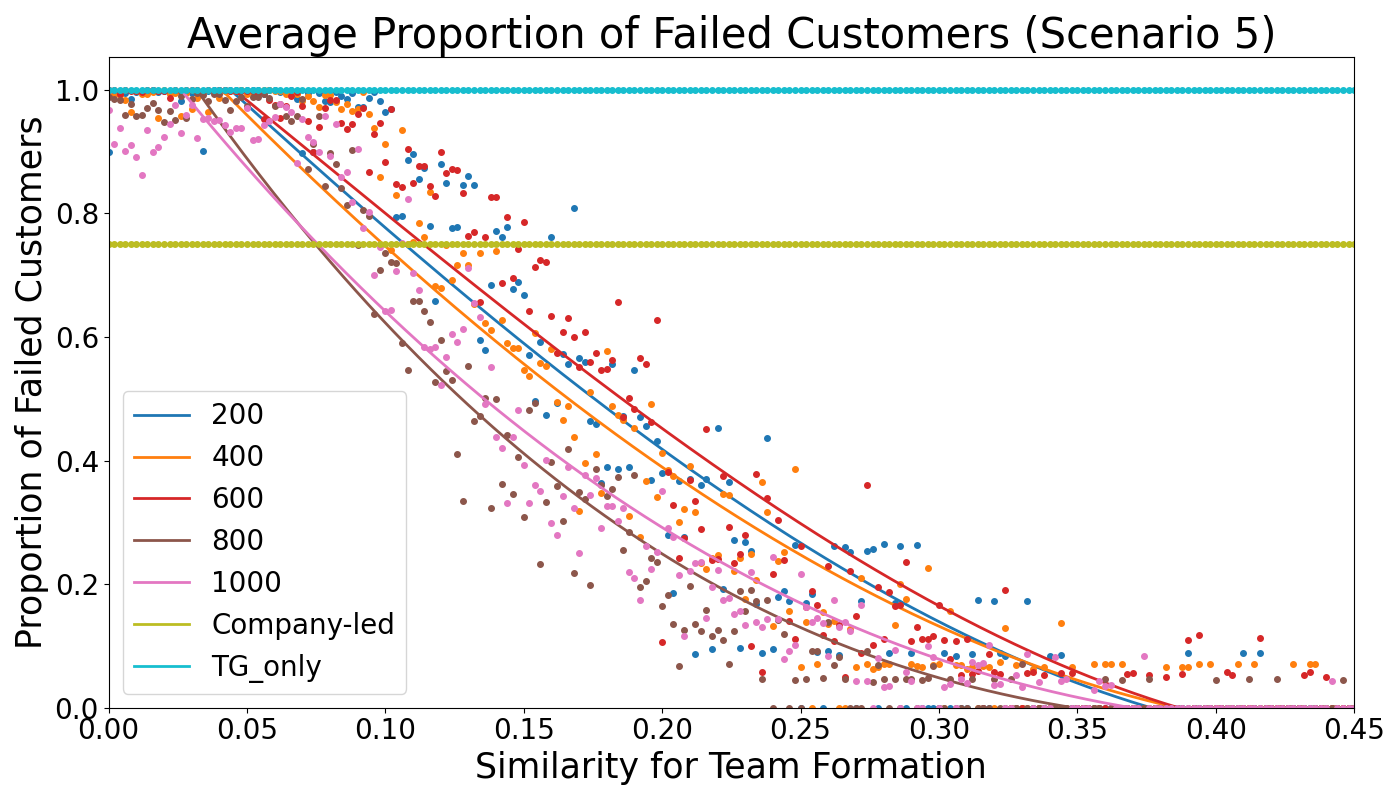}
    \caption{Average proportion of failed customers (Scenario 5)}\label{failed_customers_s5}
\end{figure}

\textbf{Fig.~\ref{failed_customers_s5}} shows the percentage of customers with uncompleted tasks. The yellow line (labelled ``Company-led") shows that approximately $75\%$ of customers failed to complete their tasks, which is significantly reduced compared with Scenario 1. It is observed that the incomplete rates for all tasks fall below the yellow line when similarity exceeds 0.08 and approach a $100\%$ success rate as similarity increases. However, at the same similarity level, a lower number of tasks correlates with a higher percentage of uncompleted tasks. For example, at a similarity of 0.2, about $50\%$ of customers in $200$ tasks remain uncompleted, compared to $30\%$ in $1000$ tasks.

\begin{figure}[ht]
    \centering
    \includegraphics[scale=0.23]{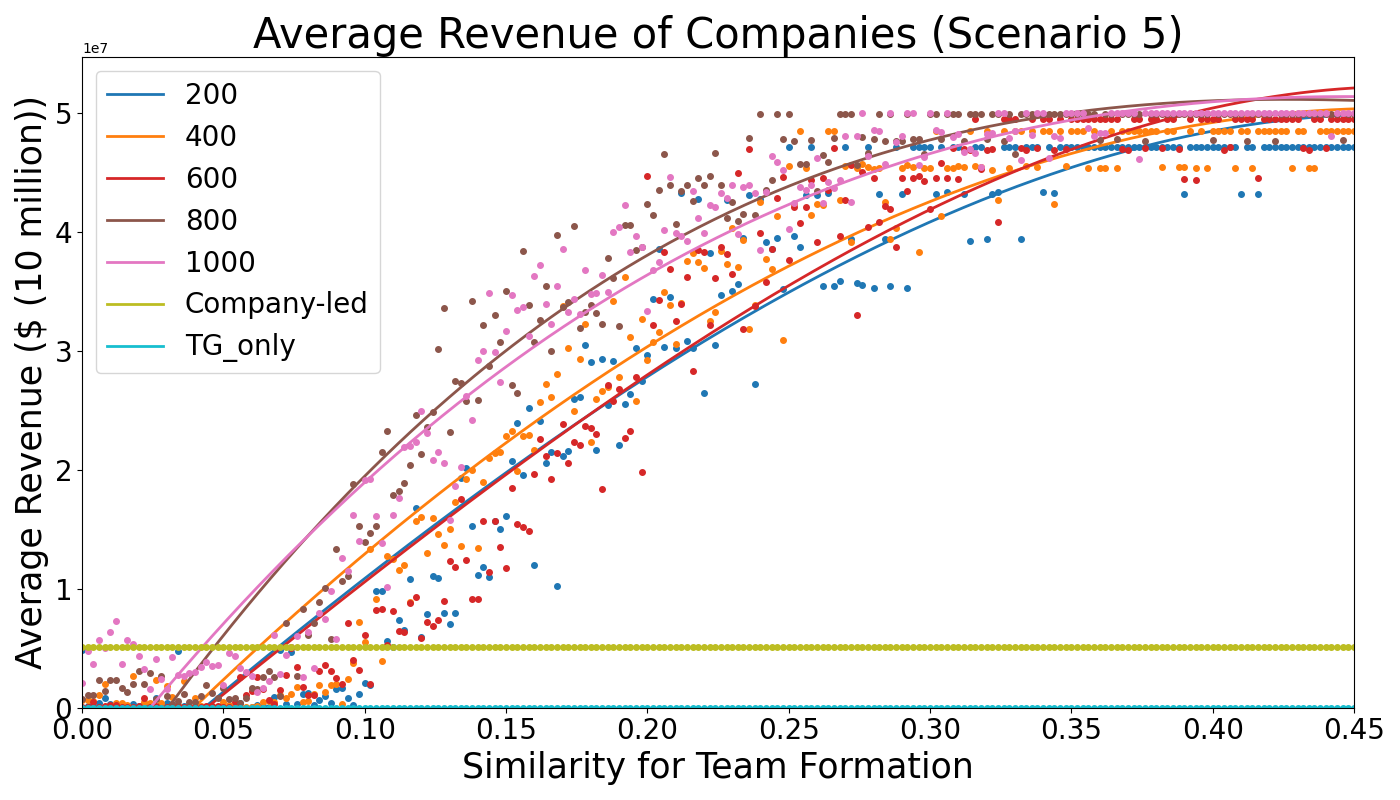}
    \caption{Average revenue of companies (Scenario 5)}\label{profit_s5}
\end{figure}

\textbf{Fig.~\ref{profit_s5}} displays polynomial fitting functions for companies' average revenue as a function of task volumes, with each scenario marked by differently coloured data points. The yellow line, labelled ``Company-led" serves as a standard reference, while the light blue line represents the scenario where customers generate tasks without company team formation. The figure demonstrates that higher similarity for team formation significantly increases average revenues for companies. Revenues exceed the reference line once similarity surpasses $0.08$. Moreover, at a fixed similarity, more tasks yield higher average revenues; for instance, at $0.2$, revenues from 1000 (around $4*10^7$) tasks are about $1.5$ times those from 400 tasks (around $2.7*10^7$). 

\textbf{Table~\ref{NE_s5}} displays the NE states of companies across various clustering scenarios. As the number of tasks increases, a lower similarity threshold is needed to reach equilibrium due to tasks spanning fewer locations, simplifying fulfilment. From a payment perspective, fewer tasks lead to reduced average customer payments by up to $19.2\%$ compared to ``Company-led", as fewer tasks boost demand for the same services within those tasks, enabling higher discounts. Consequently, this reduction in payments decreases overall market payments and company revenues. In the ``Customer-led", company revenues surge by up to $1005\%$ because previously unmanageable, high-value tasks become feasible, thus completing more lucrative tasks and substantially increasing revenues. It is worth noting that the results of Scenario 5 are similar to Scenario 1, where an increase in the number of users only causes an increase in the amount paid by the market, with the rest of the relevant parameters remaining unchanged.

\begin{table}[ht]
    \centering\footnotesize
    \setlength{\tabcolsep}{0.5pt}
    \begin{tabular}{llllll}
        \hline
        \textbf{Tasks} & \textbf{Similarity} & \textbf{Utilities} (\$) & \textbf{Company-led} & \textbf{Customer-led} & \textbf{PRF}\\
        \hline
        200 & 0.344 & Payment & 172,747 & 141,438 &  -19.2\%    \\[2pt]
            &       & Revenue  & 5,465,274 & 47,146,183 & +862\%  \\[2pt]
        400 & 0.33 & Payment & 171,981 & 145,416  & -15.4\%    \\[2pt]
            &       & Revenue & 5,031,841 & 48,472,799 & +963\%  \\[2pt]
        600 & 0.326 & Payment & 172,903 & 149,331  & -13.7\%    \\[2pt]
            &       & Revenue & 4,998,323 & 49,443,699 &  +989\% \\[2pt]
        800 & 0.314 & Payment & 172,852 & 149,559  & -13.5\%    \\[2pt]
            &       & Revenue & 4,964,407 & 49,853,104 & +1005\%    \\[2pt]
        1000& 0.3 & Payment & 172,293 & 149,938    &  -13\%  \\[2pt]
            &       & Revenue & 5,104,622 & 49,979,608  &  +979\% \\[2pt]
        \hline
    \end{tabular}
    \caption{Average customers' payment and companies' revenue at Nash equilibrium in 200-1000 tasks for scenario 5}
    \label{NE_s5}
\end{table}

\end{document}